%% file: DISC-full-paper.tex
\documentclass[a4paper,UKenglish,cleveref, autoref, thm-restate]{lipics-v2021}
\usepackage{listings,xcolor}
\usepackage[colorinlistoftodos,textsize=tiny]{todonotes}
\usepackage{subcaption}
\usepackage{algorithm}
\usepackage{algpseudocodex}
\usepackage{enumerate}
\usepackage{soul}
\usepackage{multicol}

\algnewcommand{\algorithmicstruct}{\textbf{struct}}
\algnewcommand{\Struct}[1]{\algorithmicstruct\ #1}

\newcommand{\op}[1]{{\sf #1}} % operation names

\newcommand{\Insert}{\op{PIPQ-Insert}}
\newcommand{\DeleteMin}{\op{PIPQ-DeleteMin}}
\newcommand{\Search}{\op{Search}}
\newcommand{\SearchDelete}{\op{SearchDelete}}
\newcommand{\SearchPhysDel}{\op{SearchPhysDel}}

\newcommand{\WInsert}{\op{Worker-Insert}}
\newcommand{\WDeleteMin}{\op{Worker-DeleteMin}}
\newcommand{\LInsert}{\op{L-Insert}}
\newcommand{\LDeleteMin}{\op{L-DeleteMin}}
\newcommand{\LDelete}{\op{L-DeleteMaxP}}

\title{PIPQ: Strict Insert-Optimized Concurrent Priority Queue}

\author{Olivia Grimes}{Lehigh University, Bethlehem, PA, USA}{oag221@lehigh.edu}{https://orcid.org/0009-0003-4934-5720}{}

\author{Ahmed Hassan}{Lehigh University, Bethlehem, PA, USA}{ahmed.hassan@lehigh.edu}{https://orcid.org/0000-0003-0232-305X}{}

\author{Panagiota Fatourou}{FORTH ICS, Heraklion, Greece \and University of Crete, Department of Computer Science, Heraklion, Greece}{faturu@ics.forth.gr}{https://orcid.org/0000-0002-6265-6895}{}

\author{Roberto Palmieri}{Lehigh University, Bethlehem, PA, USA}{palmieri@lehigh.edu}{https://orcid.org/0000-0002-1530-4088}{}

\authorrunning{O. Grimes, A. Hassan, P. Fatourou and R. Palmieri} %TODO mandatory. First: Use abbreviated first/middle names. Second (only in severe cases): Use first author plus 'et al.'

\Copyright{Olivia Grimes, Ahmed Hassan, Panagiota Fatourou and Roberto Palmieri} %TODO mandatory, please use full first names. LIPIcs license is "CC-BY";  http://creativecommons.org/licenses/by/3.0/

\hideLIPIcs

\ccsdesc[500]{Theory of computation~Data structures design and analysis}

\keywords{Priority Queue, Concurrent Data Structures, Synchronization}

\category{} %optional, e.g. invited paper

\relatedversion{} %optional, e.g. full version hosted on arXiv, HAL, or other respository/website
\acknowledgements{The authors would like to acknowledge Mike Spear for all his input, feedback, and support. This material is based upon work supported by the National Science Foundation under Grant No. CNS-2045976 and by the Greek Ministry of Education, Religious Affairs and Sports through the HARSH project (project no. Y$\Pi$3TA - 0560901), which is carried out within the framework of the National Recovery and Resilience Plan ``Greece 2.0'' with funding from the European Union - NextGenerationEU.}%optional

\nolinenumbers %uncomment to disable line numbering

\begin{document}

\maketitle

\begin{abstract}
This paper presents PIPQ, a strict and linearizable concurrent priority queue whose design differs from existing solutions in literature because it focuses on enabling parallelism of insert operations as opposed to accelerating delete-min operations, as traditionally done. In a nutshell, PIPQ's structure includes two levels: the worker level and the leader level. The worker level provides per-thread data structures enabling fast and parallel insertions. The leader level contains the highest priority elements in the priority queue and can thus serve delete-min operations. Our evaluation, which includes an exploration of different data access patterns, operation mixes, runtime settings, and an integration into a graph-based application, shows that PIPQ outperforms competitors in a variety of cases, especially with insert-dominant workloads.
\end{abstract}

\section{Introduction}

Concurrent data structures are at the heart of modern applications and thus, their efficient and scalable implementation is of great importance. 
Priority queues, particularly, are widely used data structures with many applications in diverse domains~\cite{CFK+23,smq,PFP21-I,PFP21-II,PFP21-III,DBLP:conf/spaa/ZhangPJ24,DBLP:conf/scoop/BenaichoucheCDCMR96,DBLP:journals/csur/GuoYYLL23,10.1145/3472456.3472463}. Applications include graph and encoding algorithms, big data analysis, task scheduling and load balancing, event-driven simulation, query optimization, databases (e.g., in transaction management, deadlock handling, process handling), 
network routing, and more.

A priority queue supports two operations: \textit{insert}, which inserts a key-value pair in the priority queue, and \textit{delete-min}, which deletes and returns the element with the smallest key (indicating the \textit{highest priority}) in the queue.

The last decade has been characterized by innovations in the design of priority queues focusing mostly on improving the performance of delete-min~\cite{linden,spray,k-lsm,multi-queues,smq,practical-scalable}, which is unquestionably a sequential bottleneck for most designs since these operations all target the same highest priority element. These innovations focus mostly on relaxing the semantics of the priority queue itself to allow for more parallelism when serving the delete-min operation. In relaxed designs, the rank of the removed element is within some defined bound of the highest priority element, though it may not be the highest priority element itself. However, relaxed semantics are not sufficient for many applications in which the single highest priority element must be removed each time. Some examples that may require strict semantics include risk management systems to ensure critical tasks are addressed first, real-time systems such as medical systems and devices, which must optimally treat patients, and in security applications in which high-priority messages or threats must be processed ahead of all other requests.

Despite the highlighted importance of such strict behavior of priority queues in many applications, the focus on relaxed priority queues in the last decade left a clear gap in literature and made the performance of state-of-the-art strict priority queues significantly inferior to their relaxed counterparts.
However, if relaxation is not a viable option, the performance of applications needing strict semantics will always be limited by the underlying priority queue.
%without further innovations, current designs inevitably lose competitiveness. 
%how to maximize performance of concurrent priority queues if relaxation is not an option
In this paper, we fill this gap as follows.
%by answering the seemingly challenging question: \textit{ How to maximize performance of concurrent priority queues if relaxation is not an option}.
%Our approach in this paper to address such a challenge can be summarized as follows.
Optimizing on delete-min in the strict setting is difficult due to its sequential nature.
Thus, while still aiming to alleviate the performance impact of the pessimistic nature of delete-min, we shift our focus to maximizing the performance and parallelism of insert operations. In fact, achieving a significant speedup for insert operations can offset the overhead of the delete-min operations. Beyond that, there is a class of applications for which the performance of insert operations is important. Examples include graph analytics~\cite{10508807}, data series analysis~\cite{CFK+23,PFP21-I,PFP21-II,PFP21-III}, and computation over streamed data~\cite{10.1145/3165266}.

The core intuition is to improve the performance of insert operations by using a set of data structures, one per thread, to avoid synchronization with other threads when the inserted element is not among those with the highest priority.
When this condition is met, which we claim to be the common case for many applications~\cite{10508807,CFK+23,PFP21-I,PFP21-II,PFP21-III,10.1145/3165266}, insertions can effectively be embarrassingly parallel. This ensures perfect scalability, Non-Uniform Memory Access (NUMA)-awareness, and near-optimal performance.

We developed this intuition into PIPQ\footnote{Pronounced ``Pip Q'', standing for Parallel Inserts Priority Queue.}, our linearizable, concurrent priority queue whose main strength is its high-performance insert operation.
PIPQ's design encompasses a hierarchical structure:
\begin{itemize}
    \item In the first level, called the \textit{worker level},
we deploy one min-heap~\cite{multiproc-prog-book} per thread. Insertions are performed to the thread's min-heap unless the thread has observed the highest priority element to be inserted.

    \item The second level, called the \textit{leader level}, collects the highest-priority elements in PIPQ, hence creating a set of candidate elements to be returned by delete-min.
\end{itemize}

Delete operations remove elements from the leader level, which contains the most minimal element(s)
inserted by each thread, and thus the most minimal element(s) in the entire structure.
To avoid degrading the performance of delete-min due to this hierarchical structure, we employ the combining technique~\cite{combining,cc-synch} to minimize contention between concurrent removals.

We develop PIPQ in C++\footnote{The source code for PIPQ is publicly available here: https://github.com/sss-lehigh/pipq} and compare its performance against two state-of-the-art strict linearizable priority queues, namely the Lind\'{e}n-Jonsson priority queue~\cite{linden}, and the Lotan-Shavit priority queue~\cite{lotan-shavit}. Experiments are conducted on an Intel server equipped with 96 cores over four processors.
%Our experimental analysis shows that the proposed approach is significantly faster than its competitors in many cases.
Applications include a microbenchmark, in which threads execute a mix of workloads; benchmarks that mimic the access patterns generated by recent applications of priority queues, in which inserting is a dominating factor; and the widely used Single-Source Shortest Path (SSSP) algorithm. 
Across experiments, our evaluation reveals that the proposed algorithm exhibits great performance benefits compared to competitors when insertions dominate the workload. Importantly, PIPQ does not sacrifice performance in workloads exhibiting an equal split of both operations, nor in delete-min heavy workloads. 

%------------------------------------------------
%------------------------------------------------
%----------------- RELATED WORK -----------------
%------------------------------------------------
%------------------------------------------------

\section{Related Work}
\label{sec:related-work}

Lotan and Shavit~\cite{lotan-shavit} were the first to use a Skiplist~\cite{pugh} as a concurrent priority queue, a design which has since largely become the standard for implementing concurrent priority queues. Notably, the Lotan and Shavit~\cite{lotan-shavit} concurrent priority queue is not linearizable, instead achieving the weaker condition of being quiescently consistent~\cite{multiproc-prog-book}.
Sundell and Tsigas~\cite{sundell-tsigas} created a lock-free, linearizable concurrent priority queue based on a skip-list.
Lind\'{e}n and Jonsson then presented their strict, linearizable priority queue algorithm~\cite{linden}; their enhancement of Lotan and Shavit's priority queue to the delete-min operation delays physical deletion of elements to improve performance. The evaluation provided by Lind\'{e}n and Jonsson shows that their priority queue consistently outperforms that of Sundell and Tsigas for various types of experiments.
Braginsky et al.~\cite{CBPQ} similarly use a skiplist-based design, modifying data-nodes to be lock-free chunks and leveraging elimination to speed up concurrent delete-mins, and inserts, which insert to the first chunk. While PIPQ is optimized to increase the parallelism of insert operations, the work in~\cite{CBPQ} still uses a non-thread-local structure but reduces synchronization through coalescing elements.
Zhang and Dechev~\cite{LF-multidem} implement a lock-free priority queue based on a multi-dimensional list. Their implementation is only quiescently consistent.
%, which is a weaker correctness condition than linearizability.
Similar to these solutions, PIPQ aims to create a strict linearizable priority queue, but with a new design which does not include use of a skip-list and aims to optimize insertions, all while achieving linearizability.

More recent work on priority queues has focused on relaxed solutions.
In these designs, the element removed by a delete-min may not be the highest priority element itself, but is within some defined bound of it.
Spraylist~\cite{spray} is a well-known relaxed priority queue, and many other relaxed solutions have been proposed~\cite{k-lsm,practical-scalable,SW17,multi-queues,smq}. Unlike these relaxed solutions that optimize for the delete-min operation, PIPQ provides a strict priority queue that optimizes for insert.

A recent related work of interest is the Stealing Multi-Queue (SMQ) priority scheduler by Postnikova et al.~\cite{smq}, which innovates on the prior Multi-Queues work~\cite{multi-queues}. SMQ uses a design similar to that of the worker level of PIPQ, such that it maintains per-thread heaps. As a result, insert operations are optimized, and, in addition to that, it relaxes the semantics of the delete-min operation with the use of so-called stealing buffers. Despite their similarities, removing the relaxation of SMQ would require additional synchronization, which is effectively the goal accomplished by PIPQ.

%------------------------------------------------
%------------------------------------------------
%----------------- OVERVIEW ---------------------
%------------------------------------------------
%------------------------------------------------

\section{Overview of PIPQ }
\label{sec:alg}

PIPQ is a linearizable priority queue that provides the traditional priority queue API: \texttt{Insert}, which inserts some key-value pair, and \texttt{DeleteMin}, which removes the element with the smallest (i.e., highest priority) key in the structure. Our priority queue supports duplicate keys and/or values.
Figure~\ref{fig:pq_model} illustrates PIPQ's architecture, including an example instantiation relevant to a specific thread (red / bolded).

\begin{figure}[h]
  \includegraphics[width=\linewidth]{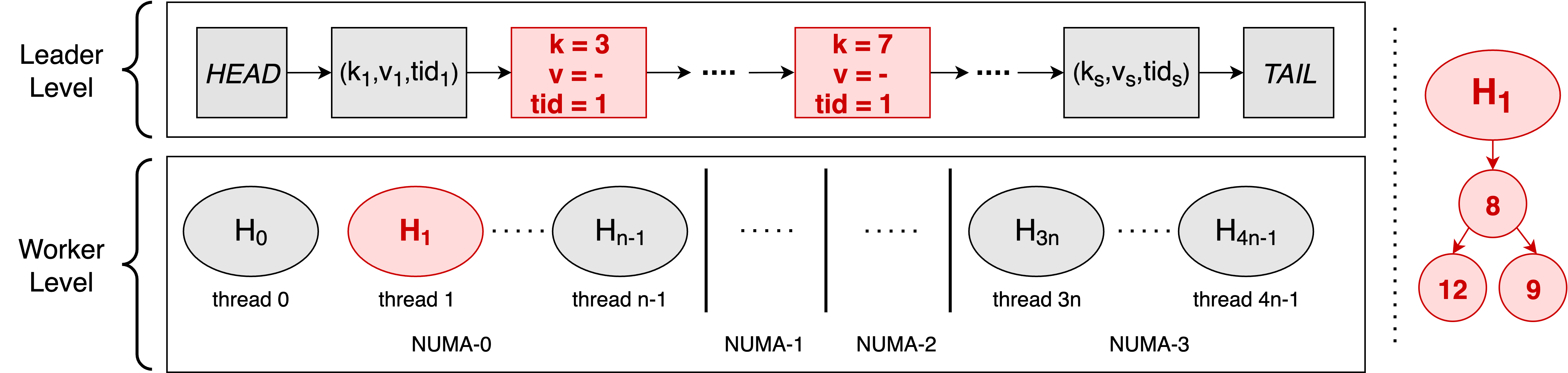}
  \caption{PIPQ on a 4-NUMA node testbed. Each oval represents a per-thread worker-level min-heap; ``n'' represents the number of threads per NUMA-node. Squares represent nodes of the leader-level linked list; ``s'' represents the size of the leader-level linked list. %NUMA nodes are shown to indicate the thread-pinning policy. Items in red (also bolded) show an example instantiation relevant to a specific thread
  }
  %, explained in Section~\ref{sec:alg}.}
  \label{fig:pq_model}
\end{figure}

To create highly parallel insertions, PIPQ maintains a two-level hierarchy of supporting data structures. The two levels include the worker level and the leader level.
The worker level provides fast, parallel insertions with minimal synchronization because each thread $p$ inserts to its own worker-level heap, $H_p$.
Each $H_p$ is implemented as a conventional min-heap protected by a single global lock. 
The simplicity of this level's design is justified by the fact that most of the time no other threads need access to $H_p$, as subsequently discussed.
At the leader level, the algorithm maintains a sorted linked list $L$ of keys, which stores the highest priority elements in PIPQ, and is meant to serve delete-min requests.

Elements may be inserted to the worker or leader level and may move between levels in order to maintain the following invariant:
%\textit{for each thread $p$, every element inserted by $p$ in $L$ is of a greater or equal priority than that of the highest priority element in the worker-level heap of $p$, $H_p$.}
\textit{for each thread $p$, every element in $L$ tagged with $p$ (that is, $p$ inserted the element into the global data structure) has greater or equal priority than the highest priority element in the worker-level heap of $p$, $H_p$.}
This ensures that when we remove the highest priority element from the leader level, it is, in fact, the highest priority element in the priority queue, allowing us to achieve linearizability without relaxing operation semantics. 
The items in red in Figure~\ref{fig:pq_model} provide an example instantiation of PIPQ relative to a specific thread, denoted thread 1.
The highest priority element in $H_1$ has priority 8. The corresponding nodes inserted by thread 1 in $L$ have a higher priority (smaller value of the key, \texttt{k}) than 8, and are ordered in $L$ by \texttt{k}.

A thread $p$ invoking insert on PIPQ will insert to the leader level only when the new element is of higher priority than the highest priority element in $H_p$ (or when $H_p$ is empty). Otherwise, the element will be inserted into $H_p$.
The leader level is intentionally kept small (explained below), so we expect most insertions to occur at the worker level.

The delete-min operation leverages the combining technique~\cite{combining,cc-synch} as used in several prior works~\cite{nvm-fifo-queue,combining-4-persistence,comb-numa-locks,comb-core}.
Electing one thread to perform a sequence of delete-min operations, which are inherently sequential, avoids the costly contention of many threads attempting to access and remove the highest priority element in the leader-level linked list.
We call this thread the \textit{coordinator}.
A thread $p$ that performs a delete-min operation on PIPQ first announces its operation and then attempts to become the \textit{leader} of its Non-Uniform Memory Access (NUMA) node. Leaders of each NUMA node compete to become the coordinator; thus, there are up to $n$ leaders competing for the coordinator role at any given time, where $n$ is the number of NUMA nodes in the machine.
The coordinator is responsible for serving the delete-min operations of threads that execute on its NUMA node.
Thus, either $p$ will become the coordinator and complete its own request along with the other requests of its NUMA node, or its request will be completed by some other thread of its NUMA node that won the race to become the coordinator.

Depending on the sequence of insert and delete-min operations, the leader-level linked list $L$ may grow and shrink arbitrarily. However, the size of $L$ is an important factor in the overall performance of PIPQ. In fact, having many elements in $L$ hinders the performance of insertions for two reasons. First, traversals on $L$ involve many elements. Second, a large size of $L$ increases the likelihood that an insertion has to act at the leader level as opposed to at the worker level, since the range of the priority of elements in $L$ increases with its size.
For these reasons, we limit the size of $L$ by requiring insertions to move elements from $L$ down to its worker-level heap if the number of elements that the thread has inserted in $L$ exceeds a certain threshold, which we call \texttt{CNTR\_MAX}.

On the other hand (and less intuitively), having too few elements in $L$ hinders the performance of delete-min operations.
To successfully linearize delete-min operations while ensuring the properties that $L$ needs to satisfy, each thread has to have at least two of the elements it has inserted in $L$; the justification for requiring at least two elements is detailed in Section~\ref{sec:leader-level-ll}. If this requirement is not met, the coordinator must promote an element from the thread's worker-level heap.
Having few elements in $L$ means that the coordinator would need to perform such a promotion more frequently. We address this issue by defining a second threshold, \texttt{CNTR\_MIN} (with a value less than that of \texttt{CNTR\_MAX}), and use it to introduce a helping mechanism.

The idea behind helping is that a thread can do work that does not directly impact the operation it is performing but can ``help'' other operations complete their work in a more efficient manner.
Since threads performing the delete-min operation must wait on the coordinator, we initiate helping during this waiting period to reduce the work required by the coordinator. Specifically, if the number of elements inserted by some thread $p$ in $L$ is less than \texttt{CNTR\_MIN}, $p$ will promote an element from $H_p$ up to the leader level instead of waiting for the coordinator to eventually do so.

%In order to easily compare to these thresholds,

We maintain an array of atomic counters containing one slot per thread such that each thread $p$'s slot in the counters array stores the number of elements currently in $L$ that were inserted by $p$ to the global structure. We refer to $p$'s counter value as $L\_count_p$.
%In summary,
The use of helping and the gap between \texttt{CNTR\_MIN} and \texttt{CNTR\_MAX} is meant to achieve the following two goals:
(1) it is less costly for a waiting thread $p$ to help by moving elements from $H_p$ to $L$, compared to the coordinator needing to do so; thus, we reduce the frequency with which the coordinator must do additional work, and 
(2) by helping only when $L\_count_p$ is less than the minimum threshold (which is, in turn, less than the maximum threshold), we reduce the frequency with which an insertion has to move an element down from the leader level to the worker level.
%Thus, using the two thresholds creates a balance that optimizes for the fast paths of our algorithm.

Before discussing the details of PIPQ, it is important to highlight that the leader-level linked list, $L$, differs from traditional linked lists in that it requires tracking which thread inserts each element in order to move them to the corresponding worker-level min-heap when necessary. To do so, $L$ stores tuples \texttt{(key,val,tid\textsubscript{p})}, where \texttt{tid\textsubscript{p}} represents the unique identifier of the thread, $p$, that first inserted this key-value pair in the priority queue. $L$ is sorted based on the key and supports the following operations:
\textit{a)} \Call{L-Insert}{key, val, tid\textsubscript{p}}, which inserts the tuple \texttt{(key,val,tid\textsubscript{p})} in $L$; 
\textit{b)} \Call{L-DeleteMin}{}, which deletes and returns the element with the smallest (highest priority) key in $L$; and
\textit{c)} \Call{L-DeleteMaxP}{lead\_largest, tid\textsubscript{p}}, which removes the element with the largest (least priority) \texttt{key} that was inserted to PIPQ by the thread with identifier \texttt{tid\textsubscript{p}} (to be subsequently inserted into $H_p$).

%------------------------------------------------
%------------------------------------------------
%------------ OPERATIONAL WORKFLOW --------------
%------------------------------------------------
%------------------------------------------------

\section{Operational Workflow}
\label{sec:apis}

Figure~\ref{fig:pq_model_combined} contains flow diagrams for the \texttt{DeleteMin} and \texttt{Insert} APIs of PIPQ. Recall the invariant maintained between levels, such that the priority of each element in $L$ must be greater than or equal to the priority of the highest priority element from their respective worker level heaps. Deciding on which level to insert, as well as the movement of elements between levels, aims to maintain this invariant. For completeness, the entire pseudocode of PIPQ is provided in Section~\ref{sec:ds-design}.

\textbf{Insertion Flow.}
We first describe how an insert operation works in PIPQ, taking into consideration the different levels to which a new element might be inserted. The following description corresponds to the diagram on the left in Figure~\ref{fig:pq_model_combined}.
The steps mentioned below refer to the numbered blocks in the diagram.

Let us consider an insert operation performed by thread $p$, identified by \texttt{tid\textsubscript{p}}, that inserts key \texttt{k} and value \texttt{v}. Per Step 1, the thread begins by locking its local, worker-level min-heap, $H_p$. In Step 2, it checks the minimum (i.e., highest priority) element in $H_p$'s key, \texttt{k\textsubscript{p}}, and compares it to \texttt{k}. This comparison determines whether the insertion will follow the so called \textit{fast path} of the algorithm or one of two slow paths, named the \textit{slower} and \textit{slowest} path. If \texttt{k} is greater than or equal to \texttt{k\textsubscript{p}} (i.e., has a lesser priority), then the fast path is taken as the insertion is thus able to be performed in a completely parallel manner. Specifically, $p$ inserts the key-value pair, \texttt{k} and \texttt{v}, into $H_p$ (Step 3), and then unlocks $H_p$ and returns (Step 11).

On the other hand, if \texttt{k} is found to be less than the minimum of $H_p$ at Step 2, then it is necessary to compare \texttt{k} to the priority of elements in the leader level in order to maintain proper ordering of elements. 
The comparisons in Steps 4 and 7 determine if the element should be inserted to the leader-level, in which case the thread follows
either the slower path (orange, middle dashed box) or the slowest path (red, right-most dashed box), or if the element should be inserted to the worker level via the fast path (following Step 7 to 3).

Recall that in order to maximize the number of insertions to a thread's local heap without affecting the efficiency of insertions to $L$, we set a maximum number of elements for $L$, \texttt{CNTR\_MAX}.
Assuming Step 4 finds that $L\_count_p$ is not equal to \texttt{CNTR\_MAX} (and is thus smaller than it), we follow the slower path.
Thus, in Step 5, we insert to $L$ using \Call{L-Insert}{} and then perform the atomic fetch-and-add operation on $L\_count_p$ (Step 6) to increment it. Finally, we unlock $H_p$ and return in Step 11.

On the contrary, if $L\_count_p$ is determined to be equal to \texttt{CNTR\_MAX} back in Step 4, this indicates that if the new element is inserted to $L$, then an element must also be moved down to the worker-level. If \texttt{k} is the largest element in $L$ tagged with \texttt{tid\textsubscript{p}}, then it would be wasted work to insert the element into $L$ just to subsequently remove it to be inserted into $H_p$. To avoid this situation, we maintain a per-thread pointer, \texttt{lead\_largest\textsubscript{p}}, which points to the node containing the largest key in $L$ tagged with some \texttt{tid\textsubscript{p}}. Step 7 compares \texttt{k} to the priority of the node pointed to by \texttt{lead\_largest\textsubscript{p}} to determine if we can insert the element to the worker level. If \texttt{k} is in fact larger or equal in priority to the current largest element (Step 7), then $p$ inserts the key-value pair, \texttt{k} and \texttt{v}, into $H_p$ (Step 3), and then unlocks $H_p$ and returns (Step 11).

\begin{figure}[h]
    %\label{fig:flows}
  \includegraphics[width=\linewidth]{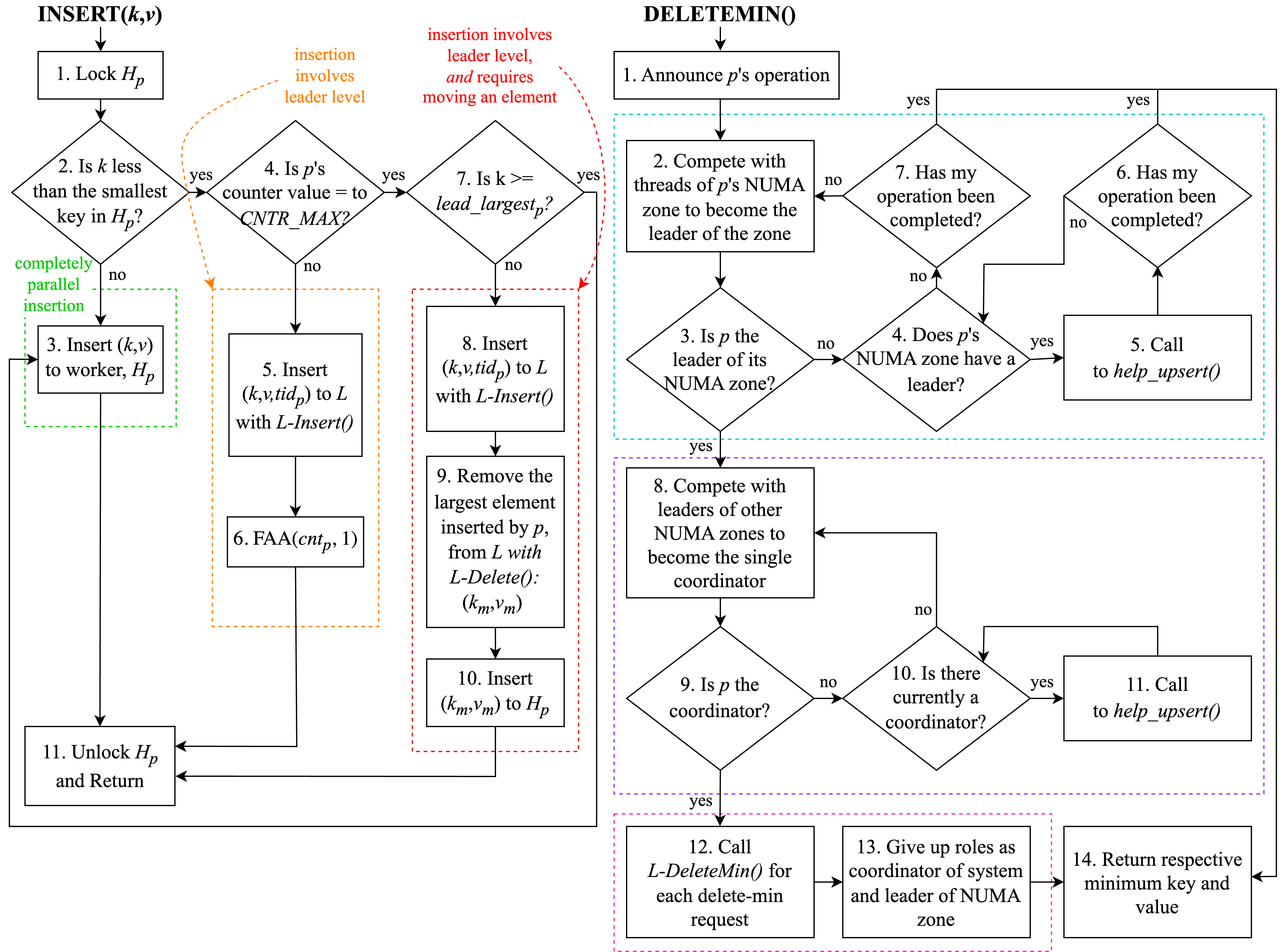}
  \caption{Flow diagrams for the Insert (left) and Delete-Min (right) operations. In each diagram, $p$ represents the thread performing the respective operation.}
  \label{fig:pq_model_combined}
\end{figure}

Otherwise, the new element must be inserted to $L$, and thus the slowest path is followed. In Step 8, \Call{L-Insert}{} is called to insert tuple (\texttt{k}, \texttt{v}, \texttt{tid\textsubscript{p}}) into $L$. Then in Step 9, \Call{L-DeleteMaxP}{} is called to remove the node pointed to by \texttt{lead\_largest\textsubscript{p}} from $L$. 
During the traversal of $L$, the algorithm tracks the relevant predecessor of \texttt{lead\_largest\textsubscript{p}} (i.e., the first node tagged with \texttt{tid\textsubscript{p}} that appears before the node pointed to by \texttt{lead\_largest\textsubscript{p}}) in order to reset the pointer. Thus, the node pointed to by \texttt{lead\_largest\textsubscript{p}} is removed from $L$, \texttt{lead\_largest\textsubscript{p}} is reset, and \Call{L-DeleteMaxP}{} returns the relevant key and value (\texttt{k\textsubscript{m}}, \texttt{v\textsubscript{m}}). Finally, (\texttt{k\textsubscript{m}}, \texttt{v\textsubscript{m}}) is inserted to $H_p$ in Step 10, and in Step 11 $H_p$ is unlocked and the operation completes.

\textbf{Delete-Min Flow.}
We now describe the flow of the delete-min operation. Steps mentioned refer to those in the right diagram of Figure~\ref{fig:pq_model_combined}.

The delete-min operation begins with thread $p$ announcing its operation in Step 1.
The blue (top-most) dashed box in the Figure shows the competition between threads of the same NUMA node to become the leader of their node. Per Steps 3 and 4, if $p$ fails to become the leader, it waits until either its operation is completed by the current coordinator (Step 6) or it eventually succeeds in becoming the leader.
While waiting in Steps 4-6, the operation calls a helper routine, \Call{Help-Upsert}{}, which checks the calling thread's counter value, $L\_count_p$, and compares it to \texttt{CNTR\_MIN}.
If the comparison finds that $L\_count_p$ is less than \texttt{CNTR\_MIN}, thread $p$ will move the highest priority element from $H_p$ to $L$, updating its counter value accordingly.

If a thread becomes the leader of its NUMA node (Step 8), then it competes to become the coordinator of the entire system, as reflected in the purple (middle) dashed box. If it does not succeed, it follows a nearly identical waiting routine as previously described, the only difference being that it does not check if its operation has been completed since the coordinator only completes operations carried by threads in its own NUMA node.

Once a thread becomes the coordinator (Step 12), it completes the delete-min operations in a combining manner~\cite{combining,cc-synch}. This is done by calling \Call{L-DeleteMin}{} for each entry in the NUMA-local announce array.
After each \Call{L-DeleteMin}{} call, it atomically decrements $L\_count_p$.
Recall that in order to maintain correctness at the leader level, each thread $p$ must have at least two elements in the linked list associated with its identifier, \texttt{tid\textsubscript{p}}.
Thus, after decrementing, if the new value of the counter is less than two, the coordinator must move an element from the relevant worker-level heap up to the leader-level. To do so, the coordinator first gets the lock on the relevant worker-level heap, $H_p$ (unless the coordinator happens to be the thread associated with identifier \texttt{tid\textsubscript{p}}). It next removes the highest priority from $H_p$, and then inserts it (augmented with \texttt{tid\textsubscript{p}}) into $L$, and finally unlocks $H_p$ if necessary (i.e., if the lock was acquired).
Before proceeding to the next active request, the coordinator informs the originating thread that its operation has been completed via the announce array.
Once the coordinator has completed all the active operations, it gives up its role as the coordinator and then as the leader of its NUMA node (Step 13) and finally returns its operation (Step 14) to the application.

%\newpage

%------------------------------------------------
%------------------------------------------------
%--------------- INTERNAL DESIGN ----------------
%------------------------------------------------
%------------------------------------------------

\section{Internal Data Structures Design}
\label{sec:ds-design}

\subsection{Structures and Definitions}

Algorithms~\ref{alg:heap-ll-impl} and \ref{alg:pipq-def} contain definitions, variables, and structures for the worker-level min-heap, the leader-level linked list, and PIPQ as a whole.
Note that not all details are provided, but what is present should be sufficient to properly understand the provided pseudocode.

% min-heap and LL defs
\begin{algorithm}[h]
\scriptsize
\caption{Worker-level min-heap and leader-level linked list definitions}\label{alg:heap-ll-impl}
\begin{algorithmic}[1]

    \LComment{Worker-level node definition}
    \State \Struct{HeapNode}
    \State \hspace{1em} $key$ : integer (the priority)
    \State \hspace{1em} $val$ : val\_t

    \Statex
    
    \LComment{Worker-level min-heap definition}
    \State \Struct{Heap}
    \State \hspace{1em} $list$ : pointer to HeapNode
    \State \hspace{1em} $size$ : integer
    \State \hspace{1em} $lock$ : pointer to integer

    \Statex

    \LComment{Leader-level node definition}
    \State \Struct{LeaderListNode}
    \State \hspace{1em} $key$ : integer (the priority)
    \State \hspace{1em} $val$ : val\_t
    \State \hspace{1em} $tid$ : integer (thread's identifier)
    \State \hspace{1em} $next$ : pointer to LeaderListNode

    \Statex

     \LComment{Leader-level linked list definition}
    \State \Struct{LeaderList}
    \State \hspace{1em} $head$ : pointer to LeaderListNode (sentinel)
    \State \hspace{1em} $tail$ : pointer to LeaderListNode (sentinel)
    \State \hspace{1em} $max\_offset$ : int
\end{algorithmic}
\end{algorithm}

% PIPQ def
\begin{algorithm}[h]
\scriptsize
\caption{PIPQ definition}\label{alg:pipq-def}
\begin{algorithmic}[1]
    
\State \Struct{PIPQ}
    
    \State \hspace{1em} $worker\_heaps[THREADS]$ : pointer to Heap
    \State \hspace{1em} $leader\_list$ : LeaderList
    \State \:
    \hspace{1em} \LComment{Locks protecting the Leaders (one per NUMA node) and the Coordinator}
    \State \hspace{1em} $compete\_coord\_locks[NUMA\_NODES]$ : pointer to integer
    \State \hspace{1em} $coord\_lock$ : pointer to integer
    
    \State \:
    
    \hspace{1em} \LComment{Additional metadata}
    \State \hspace{1em} \Struct{AnnounceStruct}
    \State \hspace{2em} $status$ : boolean
    \State \hspace{2em} $key$ : integer (priority)
    \State \hspace{2em} $val$ : val\_t
    \State \:
    \State \hspace{1em} $announce[NUMA\_NODES]$ : pointer to (array of) AnnounceStruct \Comment{For DeleteMin combining}
    \State \hspace{1em} $lead\_largest\_ptrs[THREADS]$ : pointer to LeaderListNode \Comment{Least priority element in $L$}
    \State \hspace{1em} $leader\_counters[THREADS]$ : pointer to integer \Comment{Number of elements in $L$}
    
    \State \:
    
    \hspace{1em} \LComment{The following variables (prefixed with ``$t\_$'') are all thread-local}
    \State \hspace{1em} $t\_local\_heap$ : pointer to Heap \Comment{Thread's local worker-level heap}
    \State \hspace{1em} $t\_leader\_counter$ : pointer to integer \Comment{Thread's pointer to $L\_count_p$}
    \State \hspace{1em} $t\_lead\_largest$ : pointer to LeaderListNode \Comment{Thread's least priority element in L}
    \State \hspace{1em} $t\_tid$ : constant integer \Comment{Thread's identifier, as stored in $LeaderListNode$s}
    \State \hspace{1em} $t\_compete\_coord\_lock$ : pointer to integer \Comment{Thread's NUMA-local leader lock}
    \State \hspace{1em} $t\_num\_workers$ : constant integer \Comment{Number of workers in NUMA node}
    \State \hspace{1em} $t\_announce$ : array of AnnounceStruct \Comment{NUMA-local, for DeleteMin combining}
    
    \State \:
    \hspace{1em} \LComment{User-defined config passed to PIPQ's constructor}
    \State \hspace{1em} $HLS, THREADS, CNTR\_MIN, CNTR\_MAX, MAX\_OFFSET$ : constant integer
    \State \:
    \State \hspace{1em} PIPQ($heap\_size$, $nthreads$, $cntr\_min$, $cntr\_max$, $max\_offset$) : 
    \State \hspace{2em} $HLS$($heap\_size$)
    \State \hspace{2em} $THREADS$($nthreads$)
    \State \hspace{2em} $CNTR\_MIN$($cntr\_min$)
    \State \hspace{2em} $CNTR\_MAX$($cntr\_max$)
    \State \hspace{2em} $MAX\_OFFSET$($max\_offset$) {}
\end{algorithmic}
\end{algorithm}

%------------------------------------------------
%------------------------------------------------
%----------- Leader-Level Linked List -----------
%------------------------------------------------
%------------------------------------------------

\subsection{Leader-Level Linked List}
\label{sec:leader-level-ll}

The implementation of the leader-level linked list structure, $L$, combines characteristics from two state-of-the-art lock-free algorithms, namely
Harris' sorted singly-linked list algorithm~\cite{harris} and the Lind\'{e}n-Jonsson priority queue algorithm~\cite{linden}, which implements the priority queue using a skip list; we borrow ideas from the level-zero list implementation.

The Harris work introduces the idea of marking nodes for removal before performing physical deletion (i.e., unlinking the node) to implement a lock-free linked list.
The actual technique used for marking nodes relies on stealing~\cite{DBLP:journals/pacmpl/BaudonRG23} and manipulating the last bit of the pointer\footnote{Memory is properly aligned to enable the bit-stealing technique.} pointing to the next node. The physical deletion can then be performed.

An important difference between the Harris algorithm and Lind\'{e}n-Jonsson's level-zero algorithm lies in their approach to logically marking nodes for removal before performing physical deletion. In fact, in order for the Lind\'{e}n-Jonsson algorithm to mark a node as logically removed, the preceding node's next pointer is marked instead of the next pointer of the node itself. The algorithm must mark the previous node for deletion in order to maintain a prefix of logically deleted nodes; having this prefix is an optimization that allows for delaying physical removal, which improves performance.

In our leader-level linked list implementation, we use the same optimization as the Lind\'{e}n-Jonsson algorithm in the \Call{L-DeleteMin}{} operation. However, we cannot use this marking approach for the \Call{L-DeleteMaxP}{} operation because \Call{L-DeleteMaxP}{} removes elements from the middle of the list (as opposed to the front, as in \Call{L-DeleteMin}{}).
Instead, the \Call{L-DeleteMaxP}{} operation follows the approach of Harris' algorithm and marks the next pointer of the node to be deleted.
In fact, if we were to mark the previous node's next pointer for removal, it may cause a concurrent insertion to be falsely inserted. Specifically, say a concurrent insertion inserts its new node, $n$, directly after some node $r$. It is then possible that $n$ is inserted before the physical removal of $r$, and upon the physical removal of $r$, $n$ would also be removed. This situation is prevented by marking $r$'s next pointer since this would prevent the concurrent insertion after $r$.

These two different marking approaches present a correctness issue: an \Call{L-DeleteMin}{} operation and an \Call{L-DeleteMaxP}{} operation cannot simultaneously mark the same node. Note that 
it is only possible for one \Call{L-DeleteMin}{} and one \Call{L-DeleteMaxP}{} operation on $L$ to interfere at any single time, since at each moment there is only one active thread performing \Call{L-DeleteMin}{} (i.e., the coordinator), and only one thread performing \Call{L-DeleteMaxP}{} on a specific element, since a thread $p$ only ever moves an element tagged with \texttt{tid\textsubscript{p}}.
To remedy this issue and guarantee that the \Call{L-DeleteMaxP}{} and \Call{L-DeleteMin}{} operations never interfere, we maintain at least two elements in $L$ from each thread.
For this reason, we require that \texttt{CNTR\_MIN} (and thus, \texttt{CNTR\_MAX} as well) must be at least two.

Note that we explored other designs for the leader level and found our current implementation to be the most effective. Despite the linear complexity of linked lists, our use of \texttt{CNTR\_MIN} and \texttt{CNTR\_MIN}, which reduces the likelihood that either operation must follow its worst-case path, combined with maintaining a small size of the list, makes the complexity not a practical concern. More details are provided in Section~\ref{sec:complexity}.

\subsubsection*{Implementation Details}

The leader-level linked list has two immutable sentinel nodes pointed to by \texttt{head} and \texttt{tail}, that contain $-\infty$ and $+\infty$, respectively.
The actual elements of the list are stored in nodes between the \texttt{head} and \texttt{tail} nodes. 
Each node contains a next pointer and the three fields comprising the tuple \texttt{key}, \texttt{val}, and \texttt{tid\textsubscript{p}}.

The leader-level linked list ($L$) supports the following APIs: \Call{L-Insert}{}, \Call{L-DeleteMaxP}{}, and \Call{L-DeleteMin}{}.
Internally, $L$ additionally includes helper functions \Call{search}{}, \Call{searchDelete}{}, and \Call{searchPhysDel}{}. These helper functions are all variations of the search function provided by Harris' linked list~\cite{harris}.
The implementation is lock-free, like that of Harris' linked list~\cite{harris} and Lind\'{e}n and Jonsson's priority queue~\cite{linden}.

\textbf{Helper Functions \& Marking.} Each of the three helper functions, \Call{search}{}, \Call{searchDelete}{}, and \Call{searchPhysDel}{}, must traverse the leader-level linked list by following pointers which may be marked as either \textit{logically deleted}, or \textit{moved}. We steal two bits from each node's next pointer to support the two marking techniques.
We denote the bit associated with logical deletion as the DELMIN bit, and the bit associated with moving an element as the MOVING bit. Recall that a node is considered logically deleted if the next pointer of its predecessor node's DELMIN bit is set, and moved if its own next pointer's MOVED bit is set.

The three search functions use the following methods to safely traverse the nodes and set or unset the two bits. Note that there are variants for each function, expressed between the ``\{\}'' braces, and the description of the return values for each function discusses them from left to right, respectively.
\begin{itemize}
    \item \textbf{\Call{Is\_\{ LogDel | Moving \}\_Ref}{ptr*}}: returns True/False depending on whether the LOGDEL bit or MOVING bit, respectively, is set in \texttt{ptr}.
    \item \textbf{\Call{Get\_\{ LogDel | Moving \}\_Ref}{ptr*}}: returns modified \texttt{ptr} with LOGDEL bit set or MOVING bit set, respectively.
    \item \textbf{\Call{Get\_\{ NotLogDel | Unmarked \}\_Ref}{ptr*}}: returns modified \texttt{ptr} \textit{without} the LOGDEL bit set, or neither bit set, respectively.
\end{itemize}

\begin{algorithm}[h]
\scriptsize
\caption{The \Call{L-Insert}{} method and its helping \Call{Search}{} function of the leader-level linked list.}\label{alg:l-insert}

\begin{algorithmic}[1]
\Function{L-Insert}{LeaderList $list$, LeaderListNode *$lead\_largest$, int $key$, val\_t $val$, int $tid$}
    \Repeat
        \State $l\_node, r\_node \gets \Call{search}{list, key}$
        \State $new\_node \gets \Call{NewNode}{key, val, tid, r\_node}$ \label{line:newNode}
        \If{\Call{CAS}{\&$l\_node$.next, $r\_node$, $new\_node$}} \label{line:ins_lin}
            \If{\textbf{not} *$lead\_largest$ \textbf{or} key > *$lead\_largest$.key} 
                \State *$lead\_largest$ = $new\_node$
            \EndIf
            \State \Return
        \EndIf
    \Until{True}
\EndFunction

\:

\Function{search}{LeaderList $list$, int $key$}
    \Repeat
        \LComment{1. Find l\_node and r\_node}
        \State \textcolor{red}{$prev\_log\_del \gets False$}
        \State $x \gets list$.head
        \State $x\_next \gets x$.next
        \Repeat \label{line:repeat_clause}
            \If{\textbf{not} $\Call{\textcolor{red}{is\_moving\_ref}}{x\_next}$}
                \State $l\_node \gets x$
                \State $l\_node\_next \gets \Call{\textcolor{red}{get\_notlogdel\_ref}}{x\_next}$
            \EndIf
            \State $x \gets \Call{get\_unmarked\_ref}{x\_next}$
            \If{$x = list$.tail}
                \State \textbf{break}
            \EndIf
            \State \textcolor{red}{$prev\_log\_del = \Call{is\_logdel\_ref}{x\_next}$}
            \State $x\_next \gets x$.next
        \Until{$x$.key $\geq key$ \textbf{and} \label{line:until_clause}
         \Statex \: \textbf{not} $\Call{\textcolor{red}{is\_moving\_ref}}{x\_next}$ \textbf{and}
         \Statex \: \textcolor{red}{\textbf{not} $prev\_log\_del$}}
        \State $r\_node \gets x$ \label{line:set_right}
        \State \:
        \LComment{2. Check if l\_node and r\_node are adjacent}
        \If{$l\_node\_next = r\_node$} \label{line:search_check}
            \If{\Call{\textcolor{red}{is\_logdel\_ref}}{$l\_node$.next} \textbf{or}
             \Statex \: ($r\_node$.next \textbf{and} \Call{\textcolor{red}{is\_moving\_ref}}{$r\_node$.next})}
                \State \textbf{continue}
            \EndIf
            \State \Return [$l\_node$, $r\_node$]
        \EndIf
        \State \:
        \LComment{3. Remove one or more "moving" nodes}
        \If{\Call{CAS}{\&$l\_node$.next, $l\_node\_next$, $r\_node$}} \label{line:search_cas}
            \If{\Call{\textcolor{red}{is\_logdel\_ref}}{$l\_node$.next} \textbf{or} \label{line:search_check_2}
             \Statex \: ($r\_node$.next \textbf{and} \Call{\textcolor{red}{is\_moving\_ref}}{$r\_node$.next})}
            \State \textbf{continue}
            \EndIf
            \State \Return [$l\_node$, $r\_node$]
        \EndIf
    \Until{True}
\EndFunction

\end{algorithmic}
\end{algorithm}

% L-Insert
\textbf{L-Insert.}
The \Call{L-Insert}{} operation (Algorithm~\ref{alg:l-insert}) begins by calling helper function \Call{search}{}, which locates the two nodes to insert between, namely \texttt{l\_node} and \texttt{r\_node}. The search function is very similar to that of Harris'; the differences in our implementation (red text) account for the different types of marking used by PIPQ.

In lines~\ref{line:repeat_clause}-\ref{line:until_clause},
\Call{search}{} must first pass by the prefix of logically deleted nodes, and then when finding \texttt{l\_node} and \texttt{r\_node}, account for the fact that a concurrent operation may have marked a node between the left and right nodes as moving, but not yet performed physical deletion (see the check in Line~\ref{line:search_check}). If this is the case, a compare-and-swap operation in Line~\ref{line:search_cas} of \Call{search}{} is performed to carry out the physical deletion.
Finally, \texttt{l\_node} and \texttt{r\_node} are returned to \Call{L-Insert}{}.

\Call{L-Insert}{} then allocates \texttt{new\_node} by calling \Call{NewNode}{} (Line~\ref{line:newNode}) which sets fields \texttt{key}, \texttt{val}, \texttt{tid\textsubscript{p}}, and \texttt{next}. Finally, a CAS is performed in Line~\ref{line:ins_lin} of \Call{L-Insert}{} to swing \texttt{l\_node}'s next field to point to \texttt{new\_node}. Upon a successful CAS, the algorithm returns. Otherwise, the operation is attempted again.

% L-DeleteMaxP
\textbf{L-DeleteMaxP.} Recall that the \Call{L-DeleteMaxP}{} function removes the last element in $L$ tagged with the calling thread's identifier, to be inserted to the worker level.
\Call{L-DeleteMaxP}{} (Algorithm~\ref{alg:l-delmaxp}) begins by storing a copy of \texttt{lead\_largest} in pointer \texttt{start\_lead\_largest}, to be later used when removing this node from $L$.

The algorithm then enters a loop and calls helper method \Call{searchDelete}{} in Line~\ref{line:ins_move_lr}, which identifies and returns \texttt{l\_node} and \texttt{r\_node}, such that \texttt{r\_node} is the node to be removed from $L$ to move down to the worker-level, and \texttt{l\_node} is its direct predecessor.

\Call{searchDelete}{} begins by traversing $L$ until it encounters \texttt{start\_lead\_largest} (i.e., the last node in $L$ tagged with \texttt{tid\textsubscript{p}}, the calling thread $p$'s identifier). Note that in \Call{L-DeleteMaxP}{} we must store pointer \texttt{start\_lead\_largest} (which initially points to the same node as \texttt{lead\_largest}), in case \texttt{lead\_largest} is reset in Line~\ref{line:set_last_ptr}, and then the subsequent CAS in Line~\ref{line:failed-cas} fails and the function must loop and try again.

%bbbbbbbbbbbbbb
\begin{algorithm}[H]
\scriptsize
\caption{The \Call{L-DeleteMaxP}{} method and supporting search functions of the leader-level linked list.}\label{alg:l-delmaxp}
\begin{algorithmic}[1]

\Function{L-DeleteMaxP}{LeaderList $list$, LeaderListNode *$lead\_largest$, int $tid$}
    \State $start\_lead\_largest \gets *lead\_largest$
    \Repeat
        \State $l\_ptrs \gets (lead\_largest, start\_lead\_largest)$
        \State $l\_node, r\_node \gets \Call{searchDelete}{list, l\_ptrs, tid}$ \label{line:ins_move_lr}
        \State $r\_node\_next \gets r\_node$.next
        
        \If{\textbf{not} $\Call{is\_marked\_ref}{r\_node\_next}$}
            \If{\Call{CAS}{$\&r\_node$.next, $r\_node\_next,$ \Call{get\_moving\_ref}{$r\_node\_next$}}} \label{line:ins_move_cas}
                \State \textbf{break}
            \EndIf
        \EndIf
    \Until{True}
        
    \If{\textbf{not} \Call{CAS}{\&$l\_node$.next, $r\_node$, $r\_node\_next$}} \label{line:ins_move_phys}
        \State \Call{searchPhysDel}{$list$, $r\_node$} \label{line:ins_move_phys2}
    \EndIf
    \State \Return [$r\_node$.key, $r\_node$.val]
\EndFunction

\:

\Function{searchDelete}{LeaderList $list$, LeaderListNode $start\_node$,
\Statex \: LeaderListNode $lead\_largests = (lead\_largest, start\_lead\_largest)$, int $tid$}
    \State $r\_node \gets start\_node$
    \State $x \gets start\_node$
    \State $x\_next \gets x$.next
    \Repeat
        \LComment{1. Find l\_node and r\_node}
        \While{True}
            \If{\textbf{not} $\Call{is\_moving\_ref}{x\_next}$}
                \State $cur\_l\_node \gets x$
                \State $cur\_l\_node\_next \gets \Call{get\_notlogdel\_ref}{x\_next}$
            \EndIf
            \State $x \gets \Call{get\_unmarked\_ref}{x\_next}$
            \If{$x = list$.tail}
                \State \textbf{break}
            \EndIf
            \State $x\_next \gets x$.next
            \If{\textbf{not} $\Call{is\_moving\_ref}{x\_next}$ \textbf{and} $x$.tid = $tid$}
                \State $new\_lead\_largest \gets r\_node$
                \State $l\_node = cur\_l\_node$
                \State $l\_node\_next = cur\_l\_node\_next$
                \State $r\_node = x$
                \If{$x = start\_lead\_largest$}
                    \State \textbf{break} \label{line:break-lead-larg}
                \EndIf
            \EndIf
        \EndWhile
        \State $*lead\_largest \gets new\_lead\_largest$ \label{line:set_last_ptr}
        \State \:
        \LComment{2. Check if l\_node and r\_node are adjacent}
        \If{$l\_node\_next = r\_node$} \label{line:check-phys-del-1}
            \State \Return \textit{[l\_node, r\_node]}
        \EndIf
        \State \:
        \LComment{3. Remove one or more "moving" nodes}
        \If{\Call{CAS}{\&$l\_node$.next, $l\_node\_next$, $r\_node$}} \label{line:failed-cas}
            \State \Return [$l\_node$, $r\_node$]
        \EndIf \label{line:check-phys-del-2}
        \State $x \gets list$.head
        \State $x\_next \gets x$.next
        \State $r\_node \gets new\_lead\_largest$
    \Until{1}
\EndFunction

\:

\Function{searchPhysDel}{LeaderList $list$, LeaderListNode $search\_node$}
    \Repeat
        \LComment{1. Find l\_node and r\_node}
        \State $prev\_log\_del \gets$ False
        \State $found \gets$ False
        \State $x \gets list$.head
        \State $x\_next \gets x$.next
        \Repeat
            \If{\textbf{not} $\Call{is\_moving\_ref}{x\_next}$}
                \State $l\_node \gets x$
                \State $l\_node\_next \gets \Call{get\_notlogdel\_ref}{x\_next}$
            \EndIf
            \State $x \gets \Call{get\_unmarked\_ref}{x\_next}$
            \If{$x = list$.tail} \textbf{break}
            \EndIf
            \If{$x = search\_node$} $found$ = True
            \EndIf
            \State $prev\_log\_del = \Call{is\_logdel\_ref}{x\_next}$
            \State $x\_next \gets x$.next
        \Until{$found$ = True \textbf{and} \textbf{not} $\Call{is\_moving\_ref}{x\_next}$ \textbf{and} \textbf{not} $prev\_log\_del$}
        \If{$found$ = False}
            \LComment{$search\_node$ has been removed by another thread}
            \State \Return
        \EndIf
        \State $r\_node \gets x$
        \State \:
        \LComment{2. Perform CAS to physically unlink $r\_node$ (i.e., $search\_node$)}
        \If{\Call{CAS}{\&$l\_node$.next, $l\_node\_next$, $r\_node$}}
            \State \Return
        \EndIf
    \Until{True}
\EndFunction

\end{algorithmic}
\end{algorithm}

During traversal, \Call{searchDelete}{} tracks the nearest predecessor of the node pointed to by \texttt{start\_lead\_largest} that is also tagged with \texttt{tid\textsubscript{p}} in variable \texttt{new\_lead\_largest}, in order to reset \texttt{lead\_largest}.
Once the traversal is complete (i.e., \texttt{start\_lead\_largest} has been found), \texttt{lead\_largest} is reset to \texttt{new\_lead\_largest} (Line~\ref{line:set_last_ptr} of \Call{searchDelete}{}). Note that access to \texttt{lead\_largest} is protected by the underlying lock acquired on the associated thread's heap.
As in \Call{search}{}, lines~\ref{line:check-phys-del-1}-\ref{line:check-phys-del-2} of the algorithm ensure that there are no marked nodes between \texttt{l\_node} and \texttt{r\_node} (removing any if found), before finally returning the two pointers. 

Once the left and right nodes are returned to \Call{L-DeleteMaxP}{}, a successful CAS in Line~\ref{line:ins_move_cas} indicates that the node has been logically removed from the list and the algorithm breaks from the loop; else it tries again. Lines~\ref{line:ins_move_phys}-\ref{line:ins_move_phys2} ensure that the node is physically removed from the list; if the CAS in Line~\ref{line:ins_move_phys2} fails, the search function \Call{searchPhysDel}{} is called to ensure physical removal of \texttt{r\_node}. \Call{searchPhysDel}{} varies only slightly from \Call{search}{}, such that instead of searching for a certain key value, it searches for \texttt{r\_node}. 
If the algorithm finds the node, then it will detect that it needs to perform physical removal before returning, as the other search methods do; if it is not found, then another thread has physically unlinked it.
Finally, the key and value of the element to move down to the worker-level are returned, which is subsequently handled by the thread in the calling function.

%aaaaaaaaa

\begin{algorithm}[h]
\scriptsize
\caption{The \Call{L-DeleteMin}{} method of the leader-level linked list.}\label{alg:l_del-min}
\begin{algorithmic}[1]

\Function{L-DeleteMin}{LeaderList $list$}
    \State $offset \gets 0$
    \State $x \gets list$.head
    \Repeat \label{line:delmin_repeat}
        \State $offset \gets offset + 1$
        \State $x\_next \gets x$.next
        \If{$\Call{get\_notlogdel\_ref}{x\_next} = list.$tail} \label{line:delmin_empty1}
            \State \Return EMPTY \label{line:delmin_empty2}
        \EndIf
        \If{$\Call{is\_logdel\_ref}{x\_next}$}
            \State \textbf{continue}
        \EndIf
        \State $x\_next \gets $ \Call{fetch\_and\_or}{\&$x$.next, 1} \label{line:delmin_fao}
    \Until{$x \gets \Call{get\_notlogdel\_ref}{x\_next}$ \textbf{and} \textbf{not} $\Call{is\_logdel\_ref}{x\_next}$} \label{line:delmin_until}
    \State $new\_head \gets x$
    \If{$offset > MAX\_OFFSET$} \label{line:delmin_offset}
        \State $list$.head.next $\gets$ \Call{get\_logdel\_ref}{$new\_head$} \label{line:delmin_offset_2}
    \EndIf
    \State \Return [$x$.key, $x$.val, $x$.tid]
\EndFunction

\end{algorithmic}
\end{algorithm}

% L-DeleteMin
\textbf{L-DeleteMin.}
\Call{L-DeleteMin}{} (Algorithm~\ref{alg:l_del-min}) removes the highest priority element from $L$, which is also the highest priority element in the entire data structure. 
Like the Lind\'{e}n-Jonsson algorithm, \Call{L-DeleteMin}{} performs batch physical deletions once some predetermined offset of logically deleted nodes has been reached.

The thread performing \Call{L-DeleteMin}{} first traverses past the prefix of logically deleted nodes in lines~\ref{line:delmin_repeat}-\ref{line:delmin_until}.
When the thread successfully marks the first element as logically deleted, which is verified by checking \texttt{x\_next} in Line~\ref{line:delmin_until} (i.e., the return value of the fetch-and-or from Line~\ref{line:delmin_fao}), it exits the loop.
If the offset has exceeded the maximum offset, the thread then resets the head's next pointer to the node just marked for logical deletion (lines~\ref{line:delmin_offset}-\ref{line:delmin_offset_2}).
Note that, unlike the Lind\'{e}n-Jonsson algorithm, we do not need to perform a compare-and-swap to unlink the prefix since only one thread is ever performing a delete-min operation at the time. Nonetheless, it is still beneficial for performance to delay physical deletion, as concurrent insertions face fewer cache misses when traversing the list.

\subsection{Worker-Level Min-Heap}
\label{sec:worker-heaps}

The min-heap is implemented as a pre-allocated array. If it becomes full, an additional array is allocated and linked to the first array.
Each min-heap is guarded by a single lock.
There are only ever up to two threads competing for a heap $H_p$'s lock: the thread $p$, which is local to the heap, and the thread that is the coordinator performing a delete-min operation. In practice, however, we find it very uncommon for a coordinator thread to need access to the worker level thanks to our helping mechanism, as detailed in Section~\ref{sec:alg}.

%------------------------------------------------
%------------------------------------------------
%------------ Worker-Level Min-Heap -------------
%------------------------------------------------
%------------------------------------------------

\subsection*{Implementation Details}

Algorithms~\ref{alg:worker-ins} and \ref{alg:worker-delmin} contain pseudocode for the worker-level insert and delete-min, respectively. For clarity, in the pseudocode we do not include details on handling the case in which the heap becomes full (i.e., there are more than $HLS$ elements in the heap); in the actual implementation, we support chaining heaps together by following next pointers to additional lists.

% worker insert
\begin{algorithm}[h]
\scriptsize
\caption{Inserting to a worker-level min-heap.}
\label{alg:worker-ins}
\begin{algorithmic}[1]
\Function{Worker-Insert}{Heap $heap$, int $key$, val\_t $val$}
    \If{$heap$.size $= 0$} \label{line:worker-ins-size0}
        \State $heap$.list[0] $\gets (key, val)$
        \State $heap$.size $\gets heap$.size $+ 1$
        \State \Return \label{line:worker-ins-size0-end}
    \EndIf 

    \State $idx \gets heap$.size
    \State $p\_idx \gets $ \Call{Parent}{$idx$}

    \While{$idk > 0$ \textbf{and} $key < heap$.list[$p\_idx$].key} \label{line:worker-ins-while}
        \State $heap$.list[$idx$] $\gets heap$.list[$p\_idx$]
        \State $idx \gets p\_idx$
        \State $p\_idx \gets $ \Call{Parent}{$idx$}
    \EndWhile
    \State $heap$.list[$idx$] $\gets (key, val)$ \label{line:worker-ins-set}
    \State $heap$.size $\gets heap$.size $+ 1$
    \State \textbf{return}
\EndFunction
\end{algorithmic}
\end{algorithm}

% worker delmin
\begin{algorithm}[h]
\scriptsize
\caption{Removing the highest priority element from a worker-level min-heap.}
\label{alg:worker-delmin}
\begin{algorithmic}[1]
\Function{Worker-DeleteMin}{Heap $heap$}
    \If{$heap$.size $= 0$} \label{line:worker-delmin-special0}
        \State \Return EMPTY
    \EndIf

    \If{$heap$.size $= 1$}
        \State $heap$.size $\gets 0$
        \State \Return ($heap$.list[0].key, $heap$.list[0].val) \label{line:worker-delmin-special1}
    \EndIf

    \State \:

   % \LComment{Beginning at the root, swap the last element in the heap with its left or right child until the element is properly ordered (i.e., $key$ <= keys of both children)}
     
    \State $ret \gets$ ($heap$.list[0].key, $heap$.list[0].val) \label{line:worker-delmin-read0}
    \State $(key,val) \gets heap$.list[$heap$.size $- 1$]
    \State $heap$.size $\gets heap$.size $- 1$  \label{line:worker-delmin-decsize}
    \State $idx \gets 0$ 
    \State $left\_idx$, $right\_idx \gets$ \Call{Left\_Child}{$idx$}, \Call{Right\_Child}{$idx$} \label{line:worker-delmin-read1}
    \While{($left\_idx < heap$.size \textbf{and} $key > heap$.list[$left\_idx$].key) \textbf{or} \label{line:worker-delmin-while}
        \Statex \: ($right\_idx < heap$.size \textbf{and} $key > heap$.list[$right\_idx$].key)}

        \If{$right\_idx < heap$.size \textbf{and} $heap$.list[$left\_idx$].key >= $heap$.list[$right\_idx$].key}
            \State $heap$.list[$idx$].key $\gets heap$.list[$right\_idx$].key
            \State $heap$.list[$idx$].val $\gets heap$.list[$right\_idx$].val
            \State $idx = right\_idx$
        \Else 
            \State $heap$.list[$idx$].key $\gets heap$.list[$left\_idx$].key
            \State $heap$.list[$idx$].val $\gets heap$.list[$left\_idx$].val
            \State $idx = left\_idx$
        \EndIf

        \State $left\_idx \gets$ \Call{Left\_Child}{$idx$}
        \State $right\_idx \gets$ \Call{Right\_Child}{$idx$}
    \EndWhile \label{line:worker-delmin-endwhile}

    \State \:

    \State $heap$.list[$idx$].key $\gets key$ \label{line:worker-delmin-set0}
    \State $heap$.list[$idx$].val $\gets val$ \label{line:worker-delmin-set1}
    \State \Return $ret$
\EndFunction
\end{algorithmic}
\end{algorithm}

\textbf{Worker-Insert.}
\Call{Worker-Insert}{} inserts (\texttt{key}, \texttt{val}) to the worker-level min heap specified by variable \texttt{heap} in the pseudocode.
Lines~\ref{line:worker-ins-size0}-\ref{line:worker-ins-size0-end} of Algorithm~\ref{alg:worker-ins} handle the special case in which the size of the heap is 0, in which case the first element is set to (\texttt{key}, \texttt{val}), the heap's size is increased, and then the function returns. 
Otherwise, the algorithm reads the size of the heap into variable \texttt{idx} to determine the index of the next empty spot in the heap list. The parent of \texttt{idx}, \texttt{p\_idx}, is also read so that in the loop beginning in Line~\ref{line:worker-ins-while}, \texttt{key} can be compared to the key of its parent to determine if it should move up the heap. The loop continues until the root has been reached, or \texttt{key} is larger than or equal to that of its parent, each of which indicates it is properly ordered. Finally, \texttt{key} and \texttt{val} are set at this determined location in the heap (Line~\ref{line:worker-ins-set}), its size is incremented, and the function returns.

\textbf{Worker-DeleteMin.}
\Call{Worker-DeleteMin}{} (Algorithm~\ref{alg:worker-delmin}) begins by handling special cases in which the heap's size is 0 or 1 (lines~\ref{line:worker-delmin-special0}-\ref{line:worker-delmin-special1}).
Otherwise, Line~\ref{line:worker-delmin-read0} reads the return value (i.e., the first element in the heap list) into variable \texttt{ret}, and the next line reads the \texttt{key} and \texttt{val} of the last element in the heap, to be propagated down the tree from the root until it is properly ordered in the heap. The heap's size is decremented (Line~\ref{line:worker-delmin-decsize}), and then \texttt{idx} is set to 0 (i.e., the position of the root), and its left and right children's indexes are determined so that \texttt{key} can be compared to the key of its children.
Thus, the loop in lines~\ref{line:worker-delmin-while}-\ref{line:worker-delmin-endwhile} iterates down the tree, swapping with one of \texttt{idx}'s children (specifically, the child with a smaller key) until \texttt{key} is less than or equal to both of its children, or a leaf node has been reached. Finally, the (\texttt{key},\texttt{val}) pair is set to its new position in the heap (lines~\ref{line:worker-delmin-set0}-\ref{line:worker-delmin-set1}), and the function returns.

%------------------------------------------------
%------------------------------------------------
%------------------ PIPQ APIs -------------------
%------------------------------------------------
%------------------------------------------------

\subsection{PIPQ APIs}

\textbf{Insert API.}
Algorithm~\ref{alg:pipq-insert} shows \Call{PIPQ-Insert}{}, PIPQ's Insert API. The function highly correlates to many steps of the Insertion Flow provided in Figure~\ref{fig:pq_model_combined}. Specifically, it determines to which level the element should be inserted, and if to the leader-level, additionally which path must be followed.

Thread $p$ performing \Call{PIPQ-Insert}{} to insert (\texttt{key}, \texttt{val}) begins by attempting to acquire $H_p$'s lock, $l_p$ (i.e., \texttt{t\_local\_heap.lock} in the pseudocode), in lines~\ref{line:pipq-ins-lock0}-\ref{line:pipq-ins-lock1}; this attempt is repeated until a successful CAS in Line~\ref{line:pipq-ins-lock1} indicates that the lock has been acquired.

The check in Line~\ref{line:pipq-ins-check0} indicates whether the thread must consult with the leader level. Specifically, \texttt{key} must be compared to the priority of elements in $L$ if $H_p$ is empty, or if \texttt{key} is of a higher priority than the highest priority element in $H_p$. If this is \textit{not} the case, then $p$ follows the fast-path of the algorithm (Line~\ref{line:pipq-ins-fast0}) and inserts (\texttt{key}, \texttt{val}) into $H_p$, releases $l_p$, and returns.

Otherwise, if the leader-level must be consulted, then the check in Line~\ref{line:pipq-ins-cntr-check} determines which of the two paths to follow. If $p$'s counter value, $L\_count_p$, is not equal to \texttt{CNTR\_MAX}, then it is smaller than it, and thus the slower path is followed (Line~\ref{line:pipq-ins-slower}), and so $p$ calls \Call{L-Insert}{} to insert (\texttt{key}, \texttt{val}) into $L$, increments $L\_count_p$, and then unlocks $l\_p$ and returns. If $L\_count_p$ is equal to \texttt{CNTR\_MAX}, then a final check if made which determines if the fast or slowest path will be followed.
Specifically, if \texttt{lead\_largest\textsubscript{p}} is not NULL, and the priority of the element which it points to is greater than that of \texttt{key}, then (\texttt{key}, \texttt{val}) can be safely inserted to the worker level.
Otherwise, it must be inserted into $L$, and thus
the slowest path is followed: in addition to performing \Call{L-Insert}{} to insert (\texttt{key}, \texttt{val}) into $L$, $p$ must subsequently call \Call{L-Delete}{} to move the least priority element tagged with \texttt{tid\textsubscript{p}} from $L$ to $H_p$. Finally, $p$ unlocks $l_p$ and returns.

\begin{algorithm}[h]
\scriptsize
\caption{Insert API of PIPQ, which determines which level to insert the key-value pair to (the worker-level, or the leader-level), and then calls the relevant API(s) to execute the insertion.}
\label{alg:pipq-insert}
\begin{algorithmic}[1]
\Function{PIPQ-Insert}{int $key$, key\_t $val$}
    \While{True}
        \State $lock\_val \gets t\_local\_heap$.lock \label{line:pipq-ins-lock0}
        \If{$lock\_val$ \% $2 = 0$}
            \If{CAS(\&$t\_local\_heap$.lock, $lock\_val$, $lock\_val + 1$)}  \label{line:pipq-ins-lock1}
                \LComment{Lock on local heap is acquired}
                \If{$t\_local\_heap$.size = 0 \textbf{or} $key < t\_local\_heap$.list[0].key} \label{line:pipq-ins-check0}
                    \LComment{Must compare to the leader-level}
                    \If{$t\_leader\_counter$$\rightarrow$$val = CNTR\_MAX$} \label{line:pipq-ins-cntr-check}
                        \If{$t\_lead\_largest$ \textbf{and} $key$ >= $t\_lead\_largest$$\rightarrow$$key$}
                            \LComment{Fast-path}
                            \State \Call{Worker-Insert}{$key$, $val$}
                        \Else
                            \LComment{Slowest-path}
                            \State \Call{L-Insert}{$leader$, $t\_lead\_largest$, $key$, $val$, $t\_tid$}
                            \State $(key_{\text{move}},val_{\text{move}}) = $ \Call{L-Delete}{$leader$, $t\_lead\_largest$, $t\_tid$}
                            \State \Call{Worker-Insert}{$key_{\text{move}}$, $val_{\text{move}}$}
                        \EndIf
                    \Else
                        \LComment{Slower-path}
                        \State \Call{L-Insert}{$leader$, $t\_lead\_largest$, $key$, $val$, $t\_tid$} \label{line:pipq-ins-slower}
                        \State ($t\_leader\_counter$$\rightarrow$$val) \gets (t\_leader\_counter$$\rightarrow$$val + 1)$
                    \EndIf
                \Else
                    \LComment{Fast-path}
                    \State \Call{Worker-Insert}{$key$, $val$} \label{line:pipq-ins-fast0}
                \EndIf
                \LComment{Unlock the heap}
                \State $t\_local\_heap$.lock $\gets t\_local\_heap$.lock $+ 1$
                \State \Return
            \EndIf
        \EndIf
    \EndWhile
\EndFunction
\end{algorithmic}
\end{algorithm}

\textbf{DeleteMin API.}
Algorithm~\ref{alg:PIPQ-deletemin} provides the Delete-Min API of PIPQ (\Call{PIPQ-DeleteMin}{}), as well as additional functions \Call{TryCompeteCoordinator}{}, \Call{TryBecomeCoordinator}{}, and \Call{Coordinate}{}, which compose the path a thread takes to become the Coordinator. \Call{PIPQ-DeleteMin}{} simply announces the operation, and then calls \Call{TryCompeteCoordinator}{}, in which threads belonging to the same NUMA node, $m$, compete for the NUMA-local lock, \texttt{compete\_coord\_lock\textsubscript{m}}. The thread that succeeds to acquire the lock becomes the leader of the NUMA node, and then calls \Call{TryBecomeCoordinator}{} to compete with leaders of other NUMA nodes to become the single Coordinator. Note that threads performing \Call{PIPQ-DeleteMin}{} that are not the Coordinator make calls to \Call{Help-Upsert}{} (See Algorithm~\ref{alg:help_upsert}) in lines~\ref{line:pipq-trycomp-helpup} and \ref{line:pipq-trybec-helpup} to help maintain the size of $L$.
A thread that successfully acquires \texttt{coord\_lock} becomes the Coordinator, and calls \Call{Coordinate}{}. \Call{Coordinate}{} iterates through the NUMA-local announce array, and calls \Call{Execute-Announced-DeleteMin}{} (see Algorithm~\ref{alg:do-delmin}) for each active request. Once \Call{Execute-Announced-DeleteMin}{} returns (which sets the element to be returned by the calling thread), it sets the status to False which indicates to the corresponding thread that its operation has been completed.

% delmin api + flow to coordinator role
\begin{algorithm}[h]
\scriptsize
\caption{The DeleteMin API of PIPQ, and additional functions \Call{TryCompeteCoordinator}{}, \Call{TryBecomeCoordinator}{}, and \Call{Coordinate}{}, which compose the path a thread takes to become the Coordinator.}
\label{alg:PIPQ-deletemin}
\begin{algorithmic}[1]

\LComment{PIPQ's DeleteMin API}
\Function{PIPQ-DeleteMin}{}
    \State $t\_announce[t\_tid].status = True$
    \State \Call{TryCompeteCoordinator}{}
    \State \Return $t\_announce[t\_tid].ret\_val$
\EndFunction

\State \:

\LComment{Competition between threads belonging to the same NUMA node to become its Leader}
\Function{TryCompeteCoordinator}{}
    \While{True}
        \State $lock\_val \gets *t\_compete\_coord\_lock$
        \If{$lock\_val$ \% $2 = 0$}
            \If{CAS($t\_compete\_coord\_lock$, $lock\_val$, $lock\_val+1$)}
                \LComment{Leader lock is acquired, now compete for coordinator}
                \State \Call{TryBecomeCoordinator}{}
                \State $*t\_compete\_coord\_lock \gets *t\_compete\_coord\_lock+1$
                \State \Return
            \EndIf
        \Else
            \While{$*t\_compete\_coord\_lock = lock\_val$}
                \State \Call{Help-Upsert}{} \label{line:pipq-trycomp-helpup}
                \If{\textbf{not} $t\_announce[t\_tid].status$}
                    \State \Return
                \EndIf
            \EndWhile
        \EndIf
    \EndWhile
\EndFunction

\State \:

\LComment{Competition between Leaders to become the single Coordinator}
\Function{TryBecomeCoordinator}{}
    \While{True}
        \State $lock\_val \gets *coord\_lock$
        \If{$lock\_val$ \% $2 = 0$}
            \If{CAS($coord\_lock$, $lock\_val$, $lock\_val+1$)}
                \State \Call{Coordinate}{}
                \State $*coord\_lock \gets *coord\_lock+1$
                \State \Return
            \EndIf
        \Else
            \While{$*coord\_lock = lock\_val$}
                \State \Call{Help-Upsert}{} \label{line:pipq-trybec-helpup}
            \EndWhile
        \EndIf
    \EndWhile
\EndFunction

\State \:

\LComment{The Coordinator role, performing the combining effort}
\Function{Coordinate}{}
    \For{$idx$ = 0 to $t\_num\_workers$}
        \If{$t\_announce[idx].status$}
            \State \Call{Execute-Announced-DeleteMin}{idx}
            \State $t\_announce[idx].status = False$
        \EndIf
    \EndFor
\EndFunction

\end{algorithmic}
\end{algorithm}

The final function, \Call{Execute-Announced-DeleteMin}{} is shown in Algorithm~\ref{alg:do-delmin}. The algorithm begins by calling \Call{L-DeleteMin}{} to remove the highest priority element. Assuming PIPQ is not empty (checked in Line~\ref{line:execute-emptycheck}), the relevant counter of the element just removed is retrieved, and its value decremented in lines~\ref{line:execute-getcntr}-\ref{line:execute-deccntr}. Then, \texttt{key\textsubscript{min}} and \texttt{val\textsubscript{min}} are set in the announce array to be returned to the relevant thread.

The remainder of the code (lines~\ref{line:execute-rem0}-\ref{line:execute-rem1}) handles pulling an element up from the relevant worker (that associated with \texttt{tid}) if necessary (i.e., if the counter value is less than 2, via the check in Line~\ref{line:execute-rem0}). The process follows similar logic to \Call{Help-Upsert}{} (Algorithm~\ref{alg:help_upsert}, discussed below), but may need to pull up an element from any heap as opposed to its local one. 

Lines~\ref{line:execute-lock0}-\ref{line:execute-lock1} attempt to acquire the lock on the relevant worker-level heap. If the lock is successfully acquired (indicated by a successful CAS in Line~\ref{line:execute-lock1}) then \Call{Worker-DeleteMin}{} is called to remove the highest priority element from the relevant worker heap, and then a call is made to \Call{L-Insert}{} to insert it to the leader level.
The relevant counter is incremented, the lock is released, and the function returns. It is also possible that the thread associated with \texttt{tid} concurrently upserts an element to $L$; if this is the case, then the lock acquisition will fail, and the thread will see that the counter has been incremented in Line~\ref{line:execute-nolock} and is able to return without pulling up an element itself.

\begin{algorithm}[h]
\scriptsize
\caption{The \Call{Execute-Announced-DeleteMin}{} method calls \Call{L-DeleteMin}{}, and then after modifying the relevant counter value, determines if it needs to pull an element up from the worker-level.}
\label{alg:do-delmin}
\begin{algorithmic}[1]
\Function{Execute-Announced-DeleteMin}{int $idx$}
    \State $counter\_min \gets 2$
    \State $(k_{\text{min}}, v_{\text{min}}, tid) = $ \Call{L-DeleteMin}{$leader\_list$}
    \If{$k_{\text{min}}$ \textbf{not} EMPTY} \label{line:execute-emptycheck}
        \LComment{Update the counter value for thread identifier tid, and update the announce array}
        \State $cntr\_tid = $ \Call{get\_counter}{$tid$} \label{line:execute-getcntr}
        \State \Call{add\_and\_fetch}{$cntr\_tid$, -1} \label{line:execute-deccntr}
        \State $t\_announce[idx]$.key $ = k_{\text{min}}$
        \State $t\_announce[idx]$.val $ = v_{\text{min}}$
        \State \:
        \LComment{Check if it is necessary to pull up an element from the worker-level}
        \If{$*cntr\_tid < counter\_min$} \label{line:execute-rem0}
            \State $worker\_heap \gets worker\_heap\_ptrs[tid]$
            \State $lead\_largest \gets lead\_largest\_ptrs[tid]$
            \While{True}
                \State $lock\_val \gets worker\_heap$.lock \label{line:execute-lock0}
                \If{$worker\_heap$.lock \% $2 = 0$}
                    \If{CAS(\&$worker\_heap$.lock, $lock\_val$, $lock\_val+1$)} \label{line:execute-lock1}
                        \State $(key,val) = $ \Call{Worker-DeleteMin}{$worker\_heap$}
                        \If{$*cntr\_tid = 0$}
                            \LComment{Must set the pointer to NULL if the counter becomes 0 so it is properly reset in L-Insert}
                            \State $lead\_largest$ = NULL
                        \EndIf
                        \If{$key$ \textbf{not} EMPTY}
                            \State \Call{L-Insert}{$leader$, $lead\_largest$, $key$, $val$, $tid$}
                            \State \Call{add\_and\_fetch}{$cntr\_tid$, 1}
                        \Else  
                            \State \textbf{break}
                        \EndIf
                        \State $worker\_heap$.lock $\gets worker\_heap$.lock $+ 1$
                        \State \Return
                    \EndIf
                \Else
                    \While{$lock\_val = worker\_heap$.lock}
                        \If{$*cntr\_tid >= counter\_min$} \label{line:execute-nolock}
                            \LComment{The thread identified by tid upserted an element}
                            \State \Return \label{line:execute-rem1}
                        \EndIf
                    \EndWhile
                \EndIf
            \EndWhile
        \EndIf        
    \EndIf
\EndFunction
\end{algorithmic}
\end{algorithm}

\textbf{Help-Upsert.}
The \Call{Help-Upsert}{} function (Algorithm~\ref{alg:help_upsert}) is a simple helper function, which is critical to the performance of \Call{PIPQ-DeleteMin}{} (and as a result, PIPQ as a whole).
The function, called by some thread $p$, compares its thread-local counter value, $L\_count_p$ (stored in pointer \texttt{t\_leader\_counter} in the pseudocode) to the value of \texttt{CNTR\_MIN} to determine if it should help by moving an element from its worker level $H_p$ up to the leader level.

Specifically, if $L\_count_p$ is less than \texttt{CNTR\_MIN} per the check in Line~\ref{line:help_upsert-cntrcheck}, then the thread will attempt to move an element up. Before doing so, it must acquire $l_p$ (\texttt{t\_local\_heap.lock} in the pseudocode), the lock protecting $H_p$, by checking if it is unlocked in lines~\ref{line:help_upsert-cntrcheck}-\ref{line:help_upsert-checklock} and if so, performing the CAS in Line~\ref{line:help_upsert-caslock} to attempt to acquire it. Upon a successful CAS, the lock has been acquired. The method then calls \Call{Worker-DeleteMin}{} on $H_p$, and then inserts the element returned (if it is not empty) to $L$ with \Call{L-Insert}{}. It then increments $L\_count_p$, releases $l_p$, and returns.

\begin{algorithm}[h]
\scriptsize
\caption{Helper function called by threads waiting for their delete-min operation to be completed by the Coordinator.}
\label{alg:help_upsert}
\begin{algorithmic}[1]
\Function{help-upsert}{}
    \If{$*t\_leader\_counter$ < $CNTR\_MIN$} \label{line:help_upsert-cntrcheck}
        \State $lock\_val = t\_local\_heap$.lock \label{line:help_upsert-readlock}
        \If{$lock\_val$ \% $2 = 0$} \label{line:help_upsert-checklock}
            \If{CAS(\&$t\_local\_heap$.lock, $lock\_val$, $lock\_val+1$)} \label{line:help_upsert-caslock}
                \State $(key,val) = $ \Call{Worker-DeleteMin}{$t\_local\_heap$}
                \If{$key$ \textbf{not} EMPTY}
                    \State \Call{L-Insert}{$leader$, $t\_lead\_largest$, $key$, $val$, $t\_tid$}
                    \State \Call{add\_and\_fetch}{$t\_leader\_counter$, 1}
                \EndIf
                \State $t\_local\_heap$.lock $\gets t\_local\_heap$.lock$+1$
                \State \Return
            \EndIf
        \EndIf
    \EndIf
\EndFunction
\end{algorithmic}
\end{algorithm}

%------------------------------------------------
%------------------------------------------------
%----------------- CORRECTNESS ------------------
%------------------------------------------------
%------------------------------------------------
\newtheorem{thm}{Theorem}
\newtheorem{invar}[thm]{Invariant} % invariant

\input{correctness}

%------------------------------------------------
%------------------------------------------------
%----- Performance Analysis and Discussion ------
%------------------------------------------------
%------------------------------------------------

\section{Complexity Analysis and Discussion}
\label{sec:complexity}

Given that PIPQ is blocking, and the leader-level linked list is lock-free (not wait-free), it is impossible to give a worst-case bound for either. However, for completeness, it is worth discussing the single-threaded worst-case bound, which is $O(lg(m) + k)$ for both insert and delete-min, where $m$ is the size of the relevant worker-level heap, and $k$ is the size of the leader-level linked list. This is directly inherited from the worst-case performance of its components.

We believe the above worst-case analysis is not a practical issue, due to several reasons. First, for delete-min, this worst case will only occur when the coordinator must upsert an element from a worker heap (i.e., when it removes an element and finds that its relevant counter value is less than 2). In all other cases, the coordinator removes the highest-priority element, which is the first non-logically-deleted element in the list. Also, our helping mechanism offloads most of this ``upsert'' mechanism to the pending delete-min operations that are waiting for the coordinator, which further reduces the probability that the coordinator will ever need to ``upsert'' elements. For insert, the worst case represents the scenario in which an element is inserted at the leader level, and an element is subsequently moved down to the worker level. Similar to delete-min, as shown in Figure~\ref{fig:microbench-paths}, this scenario (which we call the “slowest path”) rarely occurs in practice.

Even when this worst case scenario is unavoidable, our design mitigates its effects by configuring \texttt{CNTR\_MIN} and \texttt{CNTR\_MAX}. In particular, reducing \texttt{CNTR\_MAX} makes the size of the linked list as small as the number of threads (which is practically constant and comparable to the cost of any alternative with logarithmic-time performance). Also, by increasing \texttt{CNTR\_MIN} we increase the chance that helper threads upsert elements ahead of the coordinator, which reduces the likelihood of this worst-case scenario. In fact, some of our experimental trials that we did not include in the paper aimed at assessing the effect of replacing the leader-level linked list with other data structures, and we observed no performance gains in any of them.

%------------------------------------------------
%------------------------------------------------
%--------------------- EVAL ---------------------
%------------------------------------------------
%------------------------------------------------

\section{Evaluation}
\label{sec:eval}

We conduct our experiments on a machine running Ubuntu 22.04.3 LTS equipped with four Intel Xeon Platinum 8160 processors containing a total of 96 cores split between four NUMA nodes. Our code is written in C++ and compiled with \texttt{-std=c++20} \texttt{-O3} \texttt{-mcx16}.
Our data points represent the average of three trials performed for five seconds each. Threads are pinned to cores using a NUMA node by NUMA node policy. Hyperthreading is disabled.

We compare PIPQ with the following two state-of-the-art, strict priority queues: Lotan-Shavit~\cite{lotan-shavit} and Lind\'{e}n-Jonsson~\cite{linden}. Lotan-Shavit~\cite{lotan-shavit} is a priority queue implemented on top of Fraser's skiplist~\cite{fraser-sl}. The delete-min operation is implemented by traversing the lowest-level linked list, and atomically marking the first unmarked node as logically deleted. This data structure is not linearizable but is quiescently consistent~\cite{multiproc-prog-book}. Lind\'{e}n-Jonsson~\cite{linden} extends the Lotan-Shavit design by batching the physical deletions once a certain threshold of logically deleted nodes has been met. Lind\'{e}n-Jonsson outperforms Sundell and Tsigas' linearizable priority queue algorithm~\cite{sundell-tsigas} and Herlihy and Shavit's~\cite{multiproc-prog-book} lock-free adaptation of the LS priority queue~\cite{lotan-shavit}.

We test different applications, workloads, and data access patterns:
a traditional microbenchmark to evaluate concurrent data structures; a designated thread experiment, where each thread performs either insert \textit{or} delete-min, but not both; a phased experiment, in which the workload changes in phases; and a real application benchmark, the single-source shortest path graph algorithm.

\subsection{Microbenchmark}
\label{sec:microbench}

Results for the microbenchmark are shown in Figures~\ref{fig:microbench}, \ref{fig:microbench-lat} and \ref{fig:microbench-paths} where we evaluate throughput (millions of operations per second), latency (microseconds), and which path is followed by PIPQ's insert operations, respectively, as we increase threads. In our microbenchmark, threads perform a series of insert and delete-min operations based on the given mixed workload. The key and value are integers ranging between 1 and 100M and are randomly generated using a uniform distribution.
We include four workloads: 100\% insert, 95\% insert, 50\% insert, and 100\% delete-min operations.
Via experimentation, we find that using values of 10 for \texttt{CNTR\_MIN} and 100 for \texttt{CNTR\_MAX} achieves the goal of maintaining a proper size of the leader-level linked list, as discussed in Section~\ref{sec:alg}.

We also performed experiments (not included in the paper) in which we tracked the average number of nodes traversed by the two functions that traverse $L$ past its logically deleted prefix of nodes, namely \Call{L-Insert}{} and \Call{L-DeleteMaxP}{}, for various values of \texttt{CNTR\_MIN} and \texttt{CNTR\_MAX}.
The following discussion relates to results captured for a 50\% inserts workload.
Our results show that modifying the value of \texttt{CNTR\_MAX} has small effect on the length of traversals, whereas modifying the value of \texttt{CNTR\_MIN} has a clear correlation.
Despite this, however, the overall throughput achieved by the runs hardly varied. Thus, we conclude that the traversal on $L$ does not introduce a performance cost on PIPQ.

% throughput plots
\begin{figure}[h]
     \centering
     \begin{subfigure}{\textwidth}
         \centering
         \includegraphics[width=.4\textwidth]{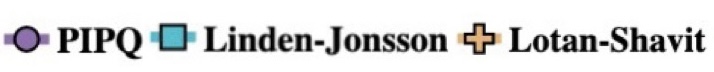}
     \end{subfigure}
     \begin{subfigure}{0.03\textwidth}
         \includegraphics[width=\textwidth]{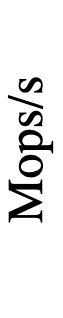}
         \vspace{20pt}
     \end{subfigure}
     \begin{subfigure}{0.22\textwidth}
         \centering
         \includegraphics[width=\textwidth]{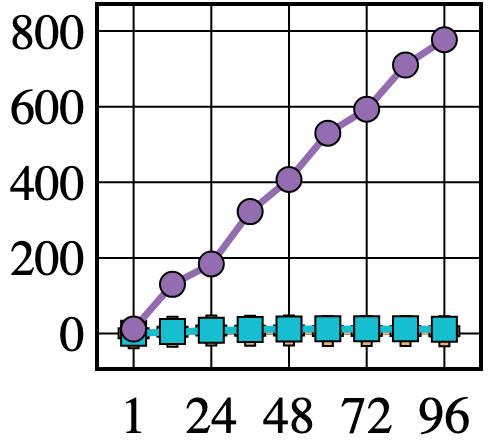}
         \caption{100\% insert}
         \label{fig:micro-100ins}
     \end{subfigure}
     \hspace{4pt}
     \begin{subfigure}{0.22\textwidth}
         \centering
         \includegraphics[width=\textwidth]{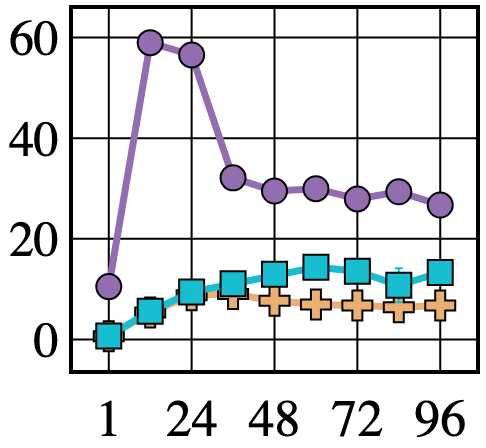}
         \caption{95\% insert}
         \label{fig:micro-95ins}
     \end{subfigure}
     \hspace{4pt}
     \begin{subfigure}{0.22\textwidth}
         \centering
         \includegraphics[width=\textwidth]{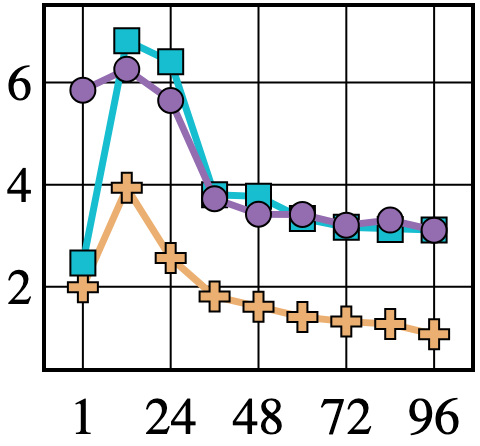}
         \caption{50\% insert}
         \label{fig:micro-50ins}
     \end{subfigure}
     \hspace{4pt}
     \begin{subfigure}{0.22\textwidth}
         \centering
         \includegraphics[width=\textwidth]{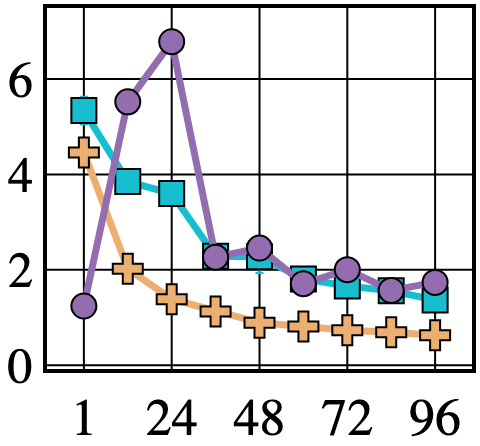}
         \caption{100\% delete-min}
         \label{fig:micro-100del}
     \end{subfigure}
     %\hfill
     \vspace{-5pt}
        \caption{Throughput using a microbenchmark. The x-axis represents the number of threads.}
    \label{fig:microbench}
\end{figure}

\textbf{Throughput Analysis.}
Refer to Figure~\ref{fig:microbench} for the results containing throughput achieved by PIPQ and its competitors.
Overall, PIPQ performs better or as well as Lind\'{e}n-Jonsson in all experiments, both of which always perform better than Lotan-Shavit.
%Spraylist performs best for delete-min heavy workloads, which is expected due to their relaxation of the delete-min operation.
When insertions dominate the workload, the parallel nature of insertions provided by PIPQ showcases the performance advantages of our approach compared to competitors.

We direct attention to Figure~\ref{fig:micro-100ins}, which shows the 100\% insert workload.
Since the leader-level linked list remains small, containing up to 100 elements per worker, and there are no threads removing the minimal keys, it becomes highly likely that threads will follow the fast path of the insert operation as more elements are inserted.
Thus, insertions become nearly embarrassingly parallel and benefit from cache locality and the speed of the sequential worker min-heap operations, providing significant speedup. At 96 threads, PIPQ performs over 76x better than its closest competitor, Lind\'{e}n-Jonsson.

Figure~\ref{fig:micro-95ins} shows the results of the 95\% insert workload. In this case, we outperform Lind\'{e}n-Jonsson by almost one order of magnitude at 12 threads (i.e., within a NUMA node) and by 2.4x at 96 threads.
To explain the performance drops in the presence of the delete-min operation, we perform an experiment to determine how many operations the coordinator performs on average for the various workloads.
We experimentally find that the coordinator performs, on average, about 24 operations in the case of the 95\% insertions workload. Thus, even with a workload of only 5\% delete-min operations, nearly every thread in a NUMA node eventually has to wait for the coordinator to complete that thread's
%its
%\textcolor{red}{each thread's}
operation (or to become the coordinator itself). Because our insert operation is so fast compared to the delete-min operation, potentially many insert requests could be served during the waiting required by the delete-min operation.

That said, however, PIPQ scales significantly better than competitors within a NUMA node, showing the algorithm's ability to increase overall performance for insert-dominant workloads despite such a pessimistic operation (i.e., delete-min) in the mix. The drop in performance after 24 threads is due to the additional synchronization needed to coordinate delete-min operations carried by threads on other NUMA nodes.
Nonetheless, PIPQ still outperforms all competitors in the 95\% insert case.

In the 50\% insert workload (Figure~\ref{fig:micro-50ins}), PIPQ performs similarly to Lind\'{e}n-Jonsson, and outperforms Lotan-Shavit. PIPQ also performs similarly to Lind\'{e}n-Jonsson for the 100\% delete-min workload (Figure~\ref{fig:micro-100del}), though exhibits performance benefits within a NUMA node (i.e., scaling up to 24 threads).
Performing so comparably to Lind\'{e}n-Jonsson for the 100\% delete-min workload for 24 threads and beyond shows an interesting correlation. Essentially, the overhead of our multi-level hierarchy, which requires moving elements between levels to ensure strict ordering, is seemingly comparable with the overhead of adjusting Lind\'{e}n-Jonsson's skip list indexing levels.

\hspace{-20pt}
\begin{minipage}{.65\textwidth}
\begin{figure}[H]
     \centering
     \begin{subfigure}{\textwidth}
         \centering
         \includegraphics[width=.6\textwidth]{plots/legend.jpg}
     \end{subfigure}\\
     %\begin{subfigure}{0.03\textwidth}
     %    \includegraphics[width=\textwidth]{plots/microsecs.png}
     %\end{subfigure}
     \begin{subfigure}{0.32\textwidth}
         \centering
         \includegraphics[width=\textwidth]{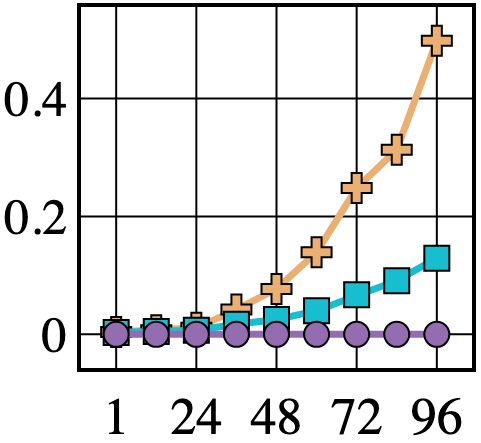}
         \caption{[ I ] 100\% insert}
         \label{fig:lat-100ins}
     \end{subfigure}
     %\hspace{-4pt}
     \begin{subfigure}{0.32\textwidth}
         \centering
         \includegraphics[width=\textwidth]{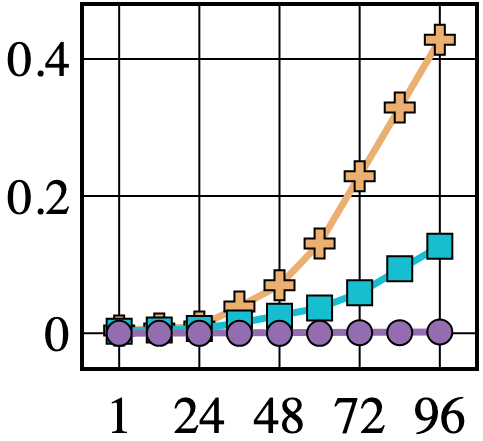}
         \caption{[ I ] 95\% insert}
         \label{fig:latins-95ins}
     \end{subfigure}
     %\hspace{-4pt}
     \begin{subfigure}{0.32\textwidth}
         \centering
         \includegraphics[width=\textwidth]{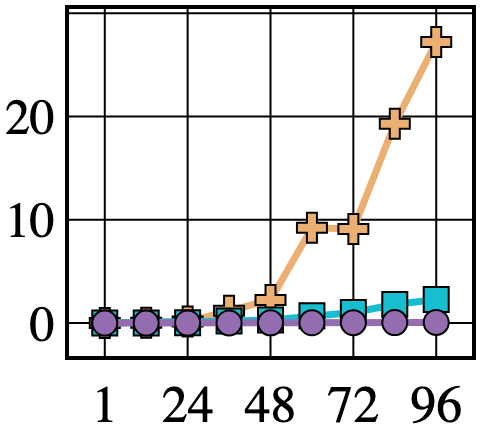}
         \caption{[ I ] 50\% insert}
         \label{fig:latins-50ins}
     \end{subfigure}\\

    %\begin{subfigure}{0.03\textwidth}
    %     \includegraphics[width=\textwidth]{plots/microsecs.png}
   %  \end{subfigure}
     \begin{subfigure}{0.3\textwidth}
         \centering
         \includegraphics[width=\textwidth]{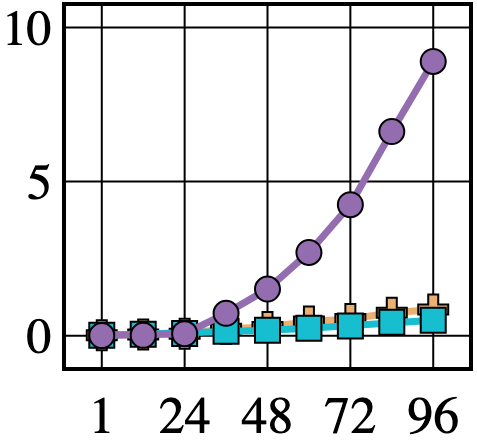}
         \caption{[ D ] 95\% insert}
         \label{fig:latdel-95ins}
     \end{subfigure}
     \hspace{4pt}
     \begin{subfigure}{0.3\textwidth}
         \centering
         \includegraphics[width=\textwidth]{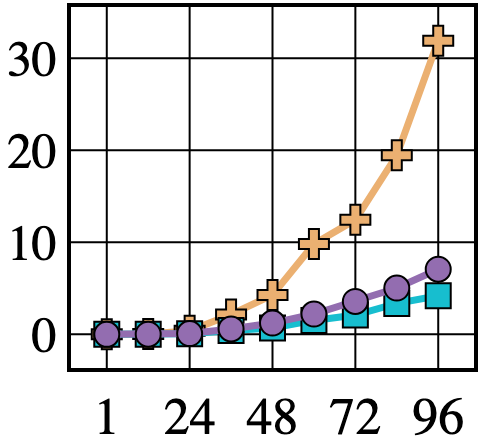}
         \caption{[ D ] 50\% insert}
         \label{fig:latdel-50ins}
     \end{subfigure}
     %\hspace{4pt}
     \begin{subfigure}{0.3\textwidth}
         \centering
         \includegraphics[width=\textwidth]{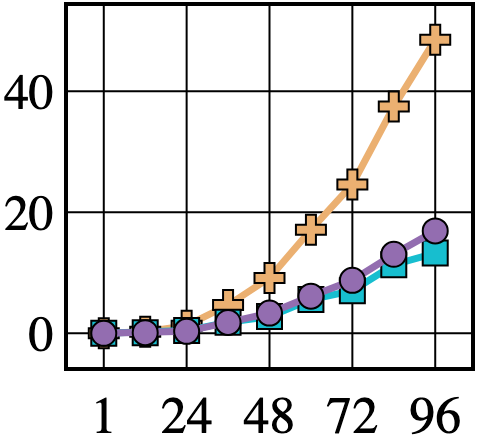}
         \caption{[ D ] 100\% del-min}
         \label{fig:lat-100del}
     \end{subfigure}
     %\hfill
        \caption{Latency using microbenchmark (less is better). [ I ] indicates insert latency, [ D ] indicates delete-min latency. \# of threads in x-axis; microseconds in y-axis.}
    \label{fig:microbench-lat}
\end{figure}
\end{minipage}
\hspace{4pt}
\begin{minipage}{.3\textwidth}
\begin{figure}[H]
     \centering
     \begin{subfigure}{\textwidth}
         \centering
         \includegraphics[width=.9\textwidth]{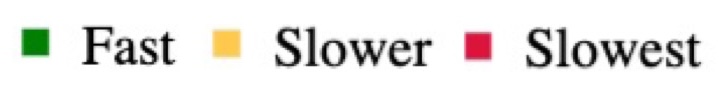}
     \end{subfigure}
     %\begin{subfigure}{0.1\textwidth}
     %    \includegraphics[width=\textwidth]{plots/perc_followed.png}
     %\end{subfigure}
     \begin{subfigure}{\textwidth}
         \centering
         \includegraphics[width=.8\textwidth]{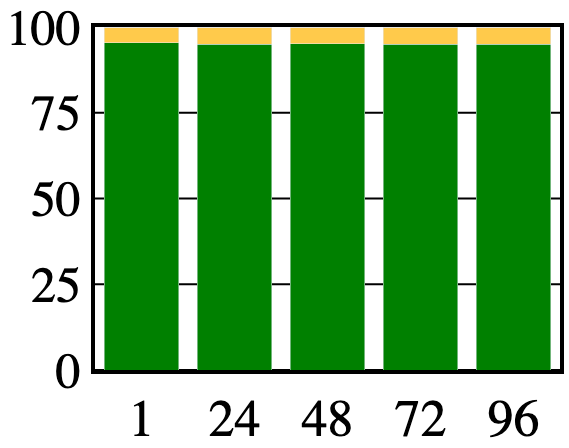}
         \caption{95\% insert}
         \label{fig:paths-95}
     \end{subfigure}\\
     \begin{subfigure}{\textwidth}
         \centering
         \includegraphics[width=0.8\textwidth]{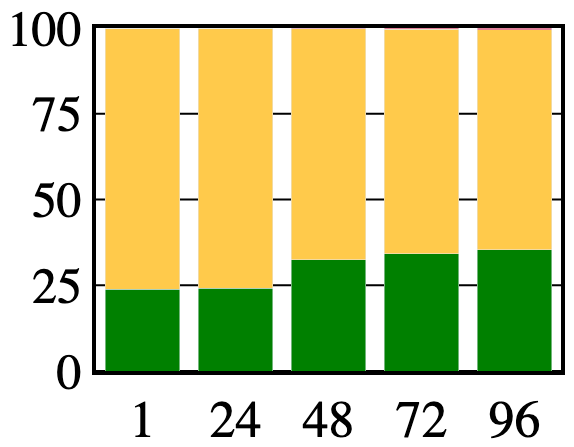}
         \caption{50\% insert}
         \label{fig:paths-50}
     \end{subfigure}
     %\vspace{-10pt}
        \caption{\% path is followed by insert operation; \# of threads in x-axis.}
    \label{fig:microbench-paths}
\end{figure}
\end{minipage}
\vspace{10pt}

%\subsubsection{Latency Analysis}
%\noindent
\textbf{Latency Analysis.}
Figure~\ref{fig:microbench-lat} shows the latency of operations for PIPQ and its competitors.
The top row of plots shows the average latency of the insert operation for relevant workloads (100\% insert, 95\% insert, and 50\% insert). In each case, PIPQ has the lowest latency compared to competitors by a large margin, revealing the benefit of PIPQ's parallel insertions.
For the 100\% insertions workload (Figure~\ref{fig:lat-100ins}) PIPQ's insert latency is 265x lower than the latency of Lind\'{e}n-Jonsson, and in the case of 95\% inserts, PIPQ's insert latency is 49x lower than that of Lind\'{e}n-Jonsson. Finally, in the case of 50\% inserts, PIPQ's latency is 40x lower than the latency of Lind\'{e}n-Jonsson.

%\hspace{5pt}
The bottom row of the plots in Figure~\ref{fig:microbench-lat} shows the average latency of the delete-min operation.
%for relevant workloads (95\% insert, 50\% insert, 100\% delete-min).
For the 50\% insert and 100\% delete-min workloads (Figures~\ref{fig:latdel-50ins} and \ref{fig:lat-100del}, respectively), PIPQ's latency is comparable to that of Lind\'{e}n-Jonsson, both of which have lower latency than Lotan-Shavit. For the 95\% insert workload on the other hand, PIPQ's latency is higher than competitors. Note that PIPQ's delete-min latency at 96 threads in the 95\% inserts workload (Figure~\ref{fig:latdel-95ins}) is very similar to the latency of delete-min at 96 threads in the 50\% workload (Figure~\ref{fig:latdel-50ins}), thus revealing that the latency of the operation remains somewhat constant across workloads. This is due to PIPQ's use of combining, and the implicit added cost of Coordinator's operating in a sequential manner.

\textbf{Path Followed by PIPQ Insert.}
In order to understand how often the benefit of inserting at the worker-level occurs for various workloads, we perform additional experiments in which we track the path (fast, slower, and slowest) followed by each insert operation.
Recall that the fast path is when an insert is to the worker-level. The slow paths involve inserting to the leader-level, and the slowest path requires the additional step of moving an element to the worker-level.
Per the results in Figure~\ref{fig:microbench-paths}, the slowest path is rarely followed (the color red is hardly visible on the plots); the percentage of the time that it is followed is less than 0.5\% for all thread counts in both workloads, significantly less so in many cases.

\vspace{5pt}
In the 95\% insert workload (Figure~\ref{fig:paths-95}), the fast path is followed approximately 95\% of the time for all thread counts, and the slower path is followed about 4.9\% of the time for all thread counts. In the case of 50\% inserts (Figure~\ref{fig:paths-50}) the fast path is taken by inserts ranging from 24\% of the time with one thread, scaling up to 35\% of the time at 96 threads, and the slower path is followed ranging from 76\% of the time at one thread, scaling down to 64\% of the time at 96 threads. The increased percentage in the slower path being followed for the 50\% insert workload compared to the 95\% insert workload is simply due to the fact that significantly more elements are being removed in the 50\% case, so the probability that an element needs to be inserted to the leader-level given a uniformly generated workload is largely increased compared to that of the 95\% case.

\subsection{Designated Thread Experiment}

Note that designating threads causes calls to \Call{Help-Upsert}{} by delete-min operations to be ineffective, since the threads performing delete-min do not ever insert. To address this issue, we move the helping effort to the insert operation, such that after performing its insert, but before it releases its lock on its local $H_p$, an insert carries out the logic of \Call{Help-Upsert}{}, removing an element from $H_p$ and inserting it into $L$ when necessary. We omit single-threaded performance as it is not a relevant data point for this experiment.

\begin{figure}[h]
    \centering
    \begin{subfigure}{\textwidth}
        % \hspace{40pt}
        \centering
        \includegraphics[width=.35\textwidth]{plots/legend.jpg}
    \end{subfigure}
    \begin{subfigure}{0.03\textwidth}
         \includegraphics[width=\textwidth]{plots/mops-s.jpg}
         \vspace{30pt}
    \end{subfigure}
    \begin{subfigure}{0.22\textwidth}
         \centering
         \includegraphics[width=\textwidth]{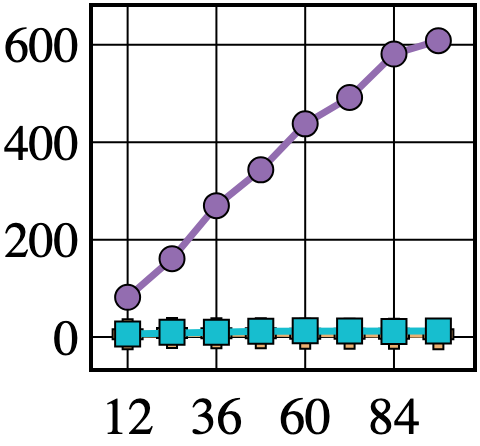}
         \caption{[ I ] d=1/24}
         \label{fig:desg-ins-4}
    \end{subfigure}
    \hspace{2pt}
    \begin{subfigure}{0.22\textwidth}
         \centering
         \includegraphics[width=\textwidth]{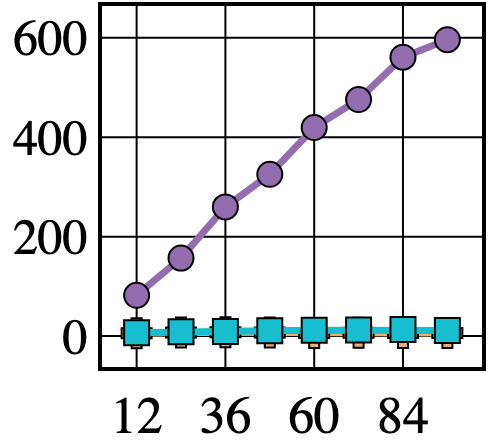}
         \caption{[ I ] d=1/12}
         \label{fig:desg-ins-8}
    \end{subfigure}
    \hspace{3pt}
    \begin{subfigure}{0.22\textwidth}
         \centering
         \includegraphics[width=\textwidth]{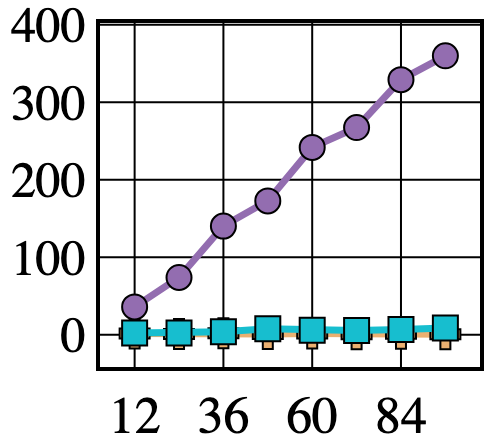}
         \caption{[ I ] d=1/2}
         \label{fig:desg-ins-48}
    \end{subfigure}
    \vspace{2pt}
    
    \begin{subfigure}{0.03\textwidth}
         \includegraphics[width=\textwidth]{plots/mops-s.jpg}
         \vspace{30pt}
    \end{subfigure}
    \begin{subfigure}{0.22\textwidth}
         \centering
         \includegraphics[width=\textwidth]{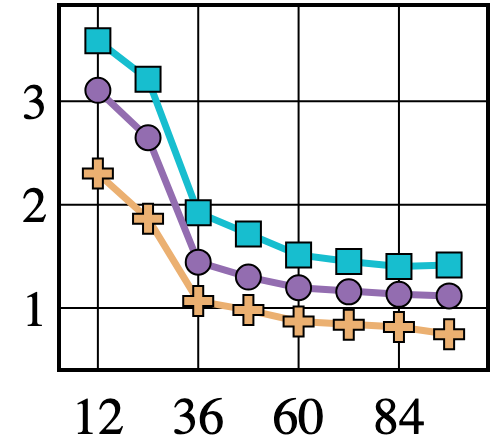}
         \caption{[ D ] d=1/24}
         \label{fig:desg-delmin-4}
    \end{subfigure}
    \hspace{2pt}
    \begin{subfigure}{0.22\textwidth}
         \centering
         \includegraphics[width=\textwidth]{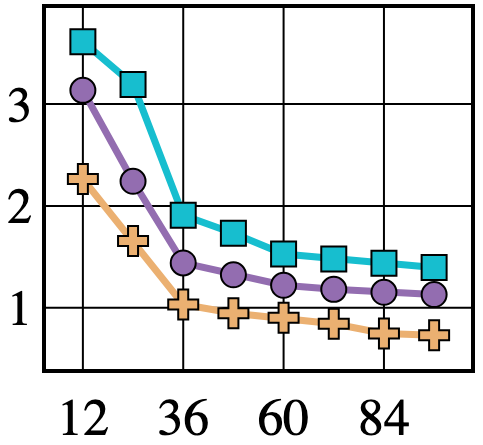}
         \caption{[ D ] d=1/12}
         \label{fig:desg-delmin-8}
    \end{subfigure}
    \hspace{3pt}
    \begin{subfigure}{0.22\textwidth}
         \centering
         \includegraphics[width=\textwidth]{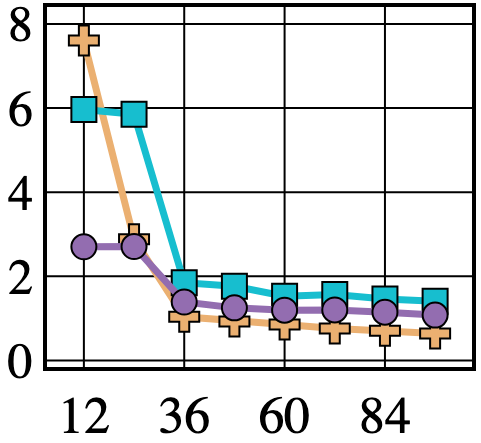}
         \caption{[ D ] d=1/2}
         \label{fig:desg-delmin-48}
    \end{subfigure}
    \vspace{-10pt}
    \caption{Throughput for the designated thread experiment. ``d'' denotes the fraction of threads performing delete-min. [ I ] indicates insert throughput, [ D ] indicates delete-min throughput. The x-axis represents the number of threads.}
    \label{fig:desg}
\end{figure}

We denote $d$ as the fraction of the total running threads that are designated to perform the delete-min operation.
We compared PIPQ with our competitors for values of $d$ = 1/24, 1/12, and 1/2. 
Making $d$ = 1/2 represents the case where half of all threads perform the delete-min operation, and the other half perform the insert operation. The values of $d$ = 1/24 and 1/12 assign more threads for insertions. The results are in Figure~\ref{fig:desg}. 
We separated the throughput of insertions and delete-mins into distinct plots to better understand the behavior of each operation.

PIPQ significantly outperforms competitors for insert throughput for all values of $d$ (see Figures~\ref{fig:desg-ins-4}, \ref{fig:desg-ins-8}, and \ref{fig:desg-ins-48}). PIPQ achieves the largest speedup for insert throughput in Figure~\ref{fig:desg-ins-4} when $d$ = 1/24, producing 50x speedup compared to the closest competitor. These experiments show that our insert operation is able to retain performance in the presence of delete-mins, as long as the threads performing insertions are not stalled due to also performing delete-min operations, as was the case in our microbenchmark experiment.
For delete-min, PIPQ achieves higher throughput than Lotan-Shavit and slightly less than that of Lind\'{e}n-Jonsson.
As a result of the high speed of inserts, when we compare the cumulative overall throughput (that is, the summation of the insert and remove-min operations), PIPQ also outperforms its closest competitor by as much as 47x.

Comparing Figures~\ref{fig:desg-delmin-48} and~\ref{fig:desg-ins-48} with Figure~\ref{fig:micro-50ins} exemplifies the effect of dedicating threads for each operation. 
In Figure~\ref{fig:micro-50ins}, it is expected that half of the threads perform delete-min, and half perform insert at any given time (similar to Figures~\ref{fig:desg-delmin-48} and~\ref{fig:desg-ins-48}).
However, PIPQ retains higher performance of insert when threads do not alternate between the two operations.

\subsection{Phased Experiment}

Some applications of priority queues involve periods of only insertions followed by periods of only delete-mins. We define this workload as ``phased'' to represent the various periods experienced by the application. 
During each phase, the total work is split among active threads. The results are shown in Figure~\ref{fig:phased}.
The experiments consist of the following phases: in each experiment, 50 million elements are inserted into the priority queue in Phase 1. Once all elements have been inserted, in Phase 2, we remove 1 million (Figure~\ref{fig:phased-1}), 5 million  (Figure~\ref{fig:phased-5}), 10 million (Figure~\ref{fig:phased-10}), and 50 million (Figure~\ref{fig:phased-50}), i.e. all, elements.

\begin{figure}[h]
    \centering
    \begin{subfigure}{\textwidth}
        \centering
        \includegraphics[width=.4\textwidth]{plots/legend.jpg}
    \end{subfigure}
    \begin{subfigure}{0.03\textwidth}
         \includegraphics[width=\textwidth]{plots/mops-s.jpg}
         \vspace{12pt}
    \end{subfigure}
    \begin{subfigure}{0.22\textwidth}
         \centering
         \includegraphics[width=\textwidth]{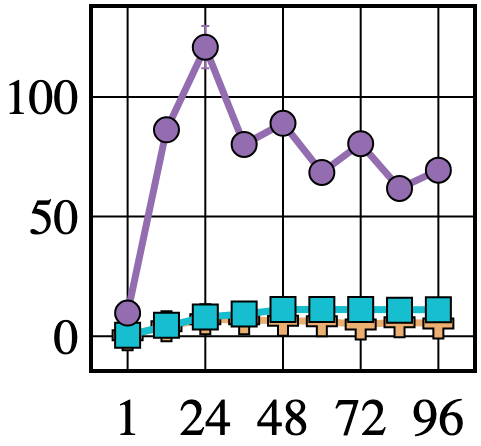}
         \caption{50->1}
         \label{fig:phased-1}
    \end{subfigure}
    \begin{subfigure}{0.22\textwidth}
         \centering
         \includegraphics[width=\textwidth]{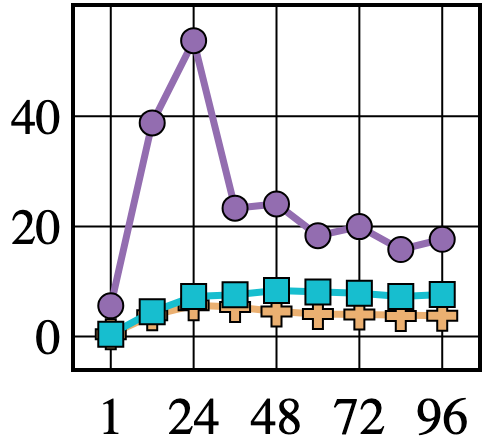}
         \caption{50->5}
         \label{fig:phased-5}
    \end{subfigure}
    \begin{subfigure}{0.22\textwidth}
         \centering
         \includegraphics[width=\textwidth]{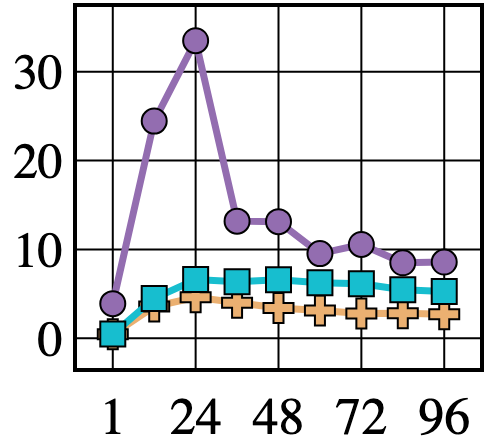}
         \caption{50->10}
         \label{fig:phased-10}
    \end{subfigure}
    \begin{subfigure}{0.22\textwidth}
         \centering
         \includegraphics[width=\textwidth]{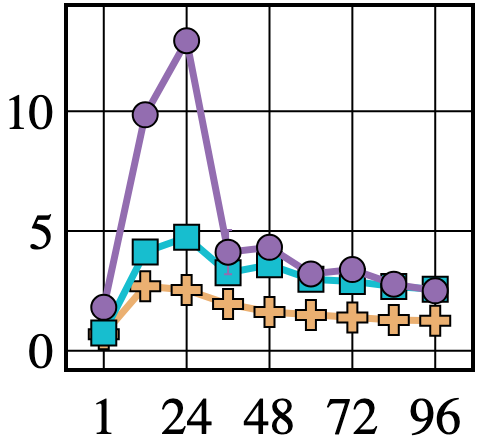}
         \caption{50->50}
         \label{fig:phased-50}
    \end{subfigure}
    \vspace{-5pt}
    \caption{Phased experiments. Workloads are denoted ``A->B'', meaning Phase 1 is A million inserts and Phase 2 is B million delete-mins. The x-axis represents threads.}
    \label{fig:phased}
\end{figure}

PIPQ outperforms competitors in the three experiments that only remove a fraction of the elements inserted, significantly so within a NUMA node (i.e., up to 24 threads). In each case, Lind\'{e}n-Jonsson is the closest competitor. As the number of delete-min operations performed in Phase 2 increases, the gap in performance becomes smaller. This aligns with one of our initial motivations that PIPQ's design fits applications where inserts outnumber deletes. In Figure~\ref{fig:phased-1}, PIPQ achieves the largest performance gains, as expected, since insertions most significantly dominate the workload in this experiment. Specifically, PIPQ produces up to 15x speedup compared to Lind\'{e}n-Jonsson at 24 threads and 6.3x speedup compared to Lind\'{e}n-Jonsson at 96 threads. In Figure~\ref{fig:phased-5}, which shows removing 10\% of inserted items, 7.4x speedup is produced compared to Lind\'{e}n-Jonsson at 24 threads, and 2.3x speedup at 96 threads. For the case of removing 20\% of the inserted items, shown in Figure~\ref{fig:phased-10}, PIPQ gains 5.1x speedup over Lind\'{e}n-Jonsson at 24 threads and 1.6x speedup at 96 threads.

In Figure~\ref{fig:phased-50}, we show when 100\% of insertions are then removed.
Similar to before, the lesser gap in performance compared to Figures~\ref{fig:phased-1}, \ref{fig:phased-5}, and \ref{fig:phased-10} is due to the asymmetry in performance between our two operations; since our insertions are very fast, the amount of time consumed during the second phase is greater compared to the first phase, and so the effect of the delete-min becomes the overwhelmingly dominating factor in performance.
That said, however, PIPQ still gains 2.7x speedup over Lind\'{e}n-Jonsson within a NUMA node.

\subsection{Single-Source Shortest Path Experiment}

\label{sec:sssp}

The Single-Source Shortest Path (SSSP) algorithm is commonly used as an application benchmark for priority queues~\cite{spray,linden,smq,multi-bucket-q}.
SSSP begins with a given source node, $s$, in a graph and finds the shortest path to all other nodes from $s$. We use the same code for the SSSP algorithm as Alhistarh et al.~\cite{spray}; as authors note, the implementation is a parallel version of Dijkstra's algorithm, which is typically used to implement sequential SSSP.

Graph applications such as SSSP,
%, which include other algorithms like
breadth-first search, the A* path-finding algorithm, and PageRank~\cite {multi-bucket-q}, compute deterministic solutions regardless of whether the supporting priority queue is relaxed or strict.
Innovations on relaxed priority queues in the past couple of decades have improved the performance of such applications (e.g., \cite{spray,smq,multi-bucket-q}).
While recognizing that strict priority queues may not be the best option for SSSP, we include the results nonetheless in order to further analyze the performance of PIPQ compared to strict competitors in a well-known benchmark.

We run the SSSP algorithm using three different datasets from \cite{stanford-datasets}: (1) a social media network, called Orkut, (2) another social media network, called LiveJournal, and (3) the California Road Network.
Our evaluation reveals an interesting correlation between the performance of PIPQ and the density of the graph, which is different in these three datasets.
%edges-to-nodes ratios, 
%The reason we selected these three datasets is because they show different edges-to-nodes ratios, 
Results are shown in Figure~\ref{fig:sssp}. Note that the y-axis represents the amount of time (in seconds) that it takes to complete the SSSP algorithm for each dataset, and thus a lesser value on the y-axis indicates better performance.

%The weighted experiments represent uniformly randomly chosen weights ranging from 1 to 100.

\begin{figure}[h]
    \centering
    \begin{subfigure}{\textwidth}
        \centering
        \includegraphics[width=.4\textwidth]{plots/legend.jpg}
    \end{subfigure}
    \begin{subfigure}{0.026\textwidth}
         \includegraphics[width=\textwidth]{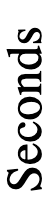}
         \vspace{30pt}
    \end{subfigure}
    \hspace{3pt}
    \begin{subfigure}{0.22\textwidth}
         \centering
         \includegraphics[width=\textwidth]{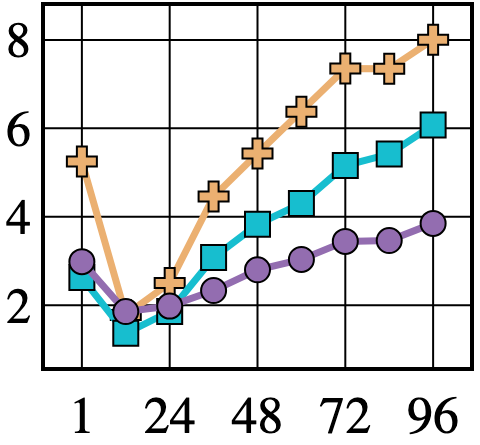}
         \caption{Orkut}
         \label{fig:orkut}
    \end{subfigure}
    \hspace{3pt}
    \begin{subfigure}{0.22\textwidth}
         \centering
         \includegraphics[width=\textwidth]{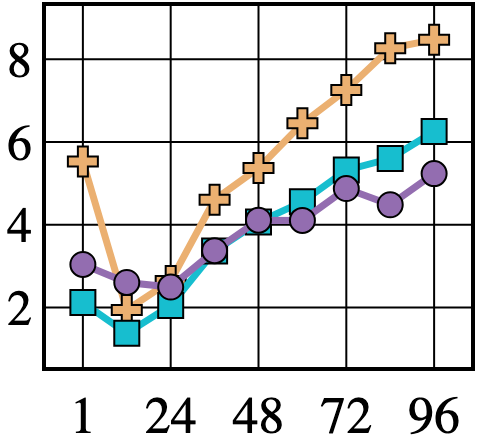}
         \caption{Live Journal}
         \label{fig:livejournal}
    \end{subfigure}
    \hspace{3pt}
    \begin{subfigure}{0.22\textwidth}
         \centering
         \includegraphics[width=\textwidth]{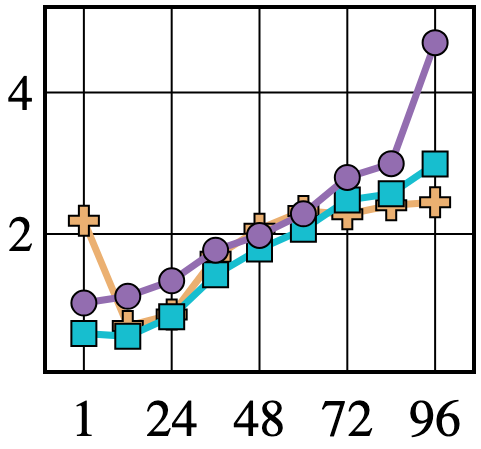}
         \caption{CA Road Network}
         \label{fig:roadnet}
    \end{subfigure}
    \vspace{-10pt}
    \caption{The Single Source Shortest Path experiments. Time in seconds is plotted; thus, \textit{a smaller y value indicates better performance.} The x-axis represents threads.}
    \label{fig:sssp}
\end{figure}

Figure~\ref{fig:orkut} contains the results of performing SSSP on the Orkut dataset. The Orkut dataset is a dense graph, containing approximately 3 million nodes and 117 million edges. Thus, there are about 39x more edges than nodes in the graph. Relative to competitors, PIPQ performs the best in this experiment compared to the other experiments. At 96 threads, PIPQ produces 1.6x speedup compared to its closest competitor, Lind\'{e}n-Jonsson.

% 6291ms (Linden) / 3879ms (PIPQ) = 1.62

Figure~\ref{fig:livejournal} contains results of performing the SSSP algorithm on the Live Journal dataset. The dataset contains approximately 5 million nodes and 70 million edges.
Similar to Orkut, though to a lesser degree, the Live Journal dataset contains many more edges than nodes; specifically, it contains 14x more edges than nodes.
PIPQ and Lind\'{e}n-Jonsson perform similarly in this experiment, with PIPQ able to retain performance better than Lind\'{e}n-Jonsson at higher thread counts.

The results of the final experiment are in Figure~\ref{fig:roadnet}.
The California Road Network dataset is smaller than Live Journal and Orkut, containing approximately 2 million nodes and 2.8 million edges, hence only 1.4x more edges than nodes.
All competitors perform similarly until 72+ threads, at which point PIPQ diverges.
%, taking about 1 second longer to complete the algorithm.

Based on these experiments, we conclude that PIPQ is best suited for datasets that contain significantly more edges than nodes, because performing SSSP on such graphs leverages the benefits provided by PIPQ's design of having data spread amongst worker-level heaps (i.e., high locality).
Since nodes are inserted and removed concurrently throughout the course of the algorithm, a smaller ratio of edges to nodes means that the number of elements in the priority queue at any given time is low as compared to a graph with a high ratio of edges to nodes, like that of Orkut and Live Journal.
This is because when there are many edges, processing a single node often leads to many subsequent insertions.
When the size of PIPQ  is large, insertions are more likely to be to the worker level, following the fast path of our algorithm.

\section{Conclusion}

In this paper we presented PIPQ, a strict, linearizable concurrent priority queue. Its design focuses on increasing the performance and parallelism of insert operations as opposed to focusing on the sequential bottleneck of delete-min operations, as done by the majority of prior work.
PIPQ achieves this by allowing insert operations to utilize per-thread data structures in the common case. Our comprehensive evaluation confirms that with PIPQ's design, it is possible to achieve high performance and scalability under various access patterns and runtime configurations.

%%
%% Bibliography
%%

%% Please use bibtex, 

\bibliography{citations}

\end{document}

%% file: correctness.tex
%------------------------------------------------
%------------------------------------------------
%----------------- CORRECTNESS ------------------
%------------------------------------------------
%------------------------------------------------

%\newpage
\section{Correctness}
\label{sec:correctness}

We assume a typical asynchronous system with $n$ threads which communicate by accessing shared variables. 
To prove the correctness of PIPQ, we break it down to its two-level components: 
the worker level and the leader level. 
At the worker level, each thread $p$ has its own min-heap $H_p$ protected by a single global lock $l_p$. 
Lock $l_p$ also protects access to $lead\_largest_p$. 
At the leader level, PIPQ maintains a linked list $L$ of keys.
Consider an execution $\alpha$ of PIPQ. 
All our claims below are in reference to this execution. 

Recall that $L$ contains two sentinel nodes, one storing the key $-\infty$ pointed to by pointer $head$,
and another storing the key $+\infty$ pointed to by the pointer $tail$. 
The {\em elements} (or {\em nodes}) of $L$ at some point $t$ are the keys (nodes) that are traversed 
if we start from the node pointed to by $head$ at $t$
and we move forward following the value (at $t$) of the next pointer of each node.
We call $p$-element an element of $L$ which is tagged with $p$ (i.e., the $tid$ field 
of the element's struct has the value $p$).
%\todo{OG: Reconcile tid vs p vs tidp in paper.}
Initially, $L$ contains just the two
sentinel nodes and it is sorted.

\vspace{5pt}
\noindent
{\bf Properties of $H_p$, for each $p \in \{1, \ldots n \}$.}
We denote by \WInsert\ (\WDeleteMin) an insert (delete-min) operation on $H_p$.
PIPQ ensures that at each point $t$, the only threads that may attempt to perform operations on $H_p$ 
or update $lead\_largest_p$ are $p$ and the  thread that plays the role of the coordinator at $t$. 
By the flow-chart of \Insert\ (Fig.~\ref{fig:pq_model_combined}), we see that a thread $p$,
maintains $l_p$ from the beginning to the end of the execution of a \Insert. 
PIPQ also ensures that during the execution of any instance of $help\_upsert$ 
(that may be invoked while $p$ is executing a \DeleteMin\ operation), 
$p$ holds $l_p$. The coordinator also acquires $l_p$ before it deletes an element from $H_p$ and releases
it after inserting it to $L$ (this is done during a \DeleteMin\ operation if the coordinator 
determines that the number of elements that are tagged with $p$ in $L$ is smaller than $2$). 
These imply that $H_p$ and $lead\_largest_p$ is accessed in an atomic way.
For simplicity of the proof, assume that  $lead\_largest_p$ is read once 
just after $l_p$ is acquired, whereas it is always updated just before the lock is released.

After acquiring its coarse-grain lock $l_p$, 
$p$ executes the serial code of \WInsert\ or \WDeleteMin\ on $H_p$. 
We linearize a \WInsert\ or a \WDeleteMin\ at the time that it acquires $l_p$.
The following lemma is an immediate consequence of the above arguments; part (\ref{largest}) 
requires additional inspection of the flow chart of \Insert\ (Steps 2, 4, 7). 

Consider any time $t$ where no Worker-operation is active on $H_p$, i.e., $H_p$ is in a quiescent state.
Then, the {\em set of elements contained in $H_p$} at $t$ is the set of keys the array implementing the min-heap contains.
PIPQ ensures that when a thread $p$ performs \LInsert\ to insert an element in $L$, it holds $l_p$
and thus no worker operation is active on $H_p$ while the \LInsert\ is performed. The same is true
when a combiner $q$ executing a \DeleteMin\ initiated by $p$ moves up an element from $H_p$ to $L$.

\begin{lemma}
\label{lem:hp}
For each thread $p$, 
\begin{enumerate}
\item \label{hp-lin} $H_p$ is a linearizable priority queue; each worker operation that is applied on $H_p$ 
is atomic (as it is applied while the thread that executes it holds the lock $l_p$);
\item All accesses to $lead\_largest_p$ are performed atomically (while lock $l_p$ is held). 
Whenever $lead\_largest_p$ is read, it points to the $p$-element with the largest key in $L$ among all $p$-elements.
\item \label{largest} While $p$ executes an \LInsert\ to insert a $p$-element in $L$, it holds the lock $l_p$; moreover, the inserted element
is of a greater or equal priority than that of the highest priority element in $p$'s worker-level heap, $H_p$, at the time
that $p$ acquired $l_p$; 
\item While $p$ inserts a $p$-element into $L$ 
it holds the lock $l_p$; 
\item \label{insert-once} Every key that appears as an argument of some \Insert\ operation invoked by $p$ (in $\alpha$) is inserted 
either in $H_p$ or in $L$ (and not in both).
\end{enumerate}
\end{lemma}

\noindent
{\bf Properties of $L$.}
We say that a node in $L$ is \textit{logically deleted} if the DELMIN bit in its predecessor's next pointer is marked (i.e., it is set to true).
A node in $L$ is \textit{moved} if the MOVING bit of its next pointer is marked. 
We say that a node is {\em marked} if either the DELMIN or the MOVING bit of its next pointer is marked. 
A node in $L$ that is neither logically deleted nor moved is {\em active}.

Lemma~\ref{lem:lcount} studies the number of (active) $p$-elements (for each $p$) that $L$ contains at each point in time. 

\begin{lemma}
\label{lem:lcount}
For each thread $p$, whenever a \LDeleteMin\ is invoked, if $H_p$ is not empty, 
then the value of $L\_count_p$ is greater than or equal to $2$
indicating that there are at least two active $p$-elements in $L$. 
\end{lemma}

\begin{proof}
Recall that $L\_count_p$ is an atomic fetch-and-add object. 
At most two processes may access $L\_count_p$ at each point in time, namely $p$ (when it inserts to the leader-level) and
the current coordinator (which may or may not be $p$).
%Specifically, $L\_count_p$ may be updated by either $p$ when it inserts to the leader-level, or by the coordinator thread.

During a \Insert, 
$p$ either increments  $L\_count_p$ (if it follows the slower path and inserts the new key in $L$), 
or leaves $L\_count_p$ unchanged (if it follows the fast path since it does not access $L$, or if it follows the slowest path since after it inserts an element in $L$, it also deletes another one from it,
thus eliminating the need for updating $L\_count_p$).

%Assume now that $p$ acts as the coordinator serving a \DeleteMin\ operation of some thread $q$.
Assume now that some thread $q$ acts as the coordinator and removes a $p$-element from $L$ (and thus decrements $L\_count_p$).
PIPQ ensures that if the new value of $L\_count_p$ (after decrementing it) is less than $2$, then $q$ will not complete the \LDeleteMin\ until the highest priority element of $H_p$ has been moved into $L$ (unless $H_p$ is empty).
Thus, either $q$ will succeed at acquiring lock $l_p$, and perform a \WDeleteMin\ on $H_p$, and then insert the returned element to the leader-level with \LInsert, or thread $p$ itself (which is thus waiting for its delete-min operation to complete) will perform \op{help\_upsert} and move an element up on behalf of $q$.
If $q$ moves the element up, then it increments $L\_count_p$, else $p$ does and $q$ sees the increment and returns.
Thus, $q$ may cause $L\_count_p$ to have a value lower than $2$ only while it serves a delete-min request of $p$, and will ensure the value is incremented back to be at least $2$ (assuming $H_p$ is non-empty) before completing the operation.
Since combining ensures that there is only one thread executing \LDeleteMin\ at each point in time, the claim follows.
\end{proof}

$L$ is implemented in a way that respects certain properties 
proved for the Lind\'{e}n-Jonsson {\em et al.} algorithm in~\cite{linden}
and the Harris algorithm in~\cite{harris}. 
Recall that the only differences between our various search functions (\Search, \SearchDelete, and \SearchPhysDel) and that of Harris account for the different types of marking used by PIPQ, and additionally for the latter two functions, what exactly is being searched for (a certain key, or a specific node via passing a pointer to it).
Such differences cannot jeopardize any of the properties of the Harris' \Search.

\begin{lemma}[{\bf \cite{harris}, Sections 4.1 and 5.1}]
\label{lem:search properties}
Consider any instance $I$ of \Search\ with key $k$. 
Let $l\_node$ and $r\_node$ be the pointers returned by $I$; call {\em left node} and {\em right node} the nodes pointed to by these pointers, respectively.
Then, there is some point $t$ during the invocation and the response of $I$, at which all the following conditions hold:
\begin{enumerate}
\item the left node and the right node are nodes of $L$ at $t$,
\item the key of the left node must be less than $k$ (unless $k$ is the highest priority in $L$, in which case the left node may be a logically deleted node whose key is greater than or equal to $k$) and the key of the right node must be greater than or equal to $k$,  
\item neither node is marked as moved and the right node is not logically deleted, and
\item the right node must be the immediate successor of the left node in $L$. 
\end{enumerate}
\end{lemma}

The next invariant states that $L$ is a list (i.e., it does not contain any cycles),
whose last element is the node pointed to by $tail$. Its proof is based on Lemma~\ref{lem:search properties}
and requires inspection of the steps where changes to the next pointers of the nodes of $L$ are performed
to argue that none of these steps violates the invariant (in a way similar as Insert and Delete 
does not corrupt the Harris' list in~\cite{harris} or DeleteMin does not
corrupt the Lind\'{e}n-Jonsson list in~\cite{linden}). 
%\y{For each of these steps the argument
%is the same as for Insert or Delete in Harris' algorithm, 
%or as for \LDeleteMin\ in Lind\'{e}n-Jonsson algorithm.} 

\begin{invar}
\label{lem:prefix}
$L$ is a list whose last element is the sentinel node pointed to by $tail$.
\end{invar}

Recall that it is only possible for one \LDeleteMin\ and one \LDelete\ operation 
on $L$ to interfere at any single time, since there is only ever one active thread performing \LDeleteMin\ at a time (i.e., the coordinator), and only one thread performing \LDelete\ on a specific element at a time.
%Lemma~\ref{lem:lcount} and Lemma~\ref{lem:hp}(\ref{largest}) imply that 
%an \LDelete\ and an \LDeleteMin\ operation will never interfere when updating $L$. 
%In more detail, 
%Lemma~\ref{lem:lcount} implies that if some $H_p$ is not empty, 
%then we have at least two $p$-elements in $L$. 
\LDeleteMin\ removes the element with the smallest key in $L$, 
whereas \LDelete\ by some process $p$ removes the $p$-element with the largest key in $L$, 
which is pointed to by $lead\_largest_p$ (by Lemma~\ref{lem:hp}(\ref{largest})).
If $H_p$ is non-empty, Lemma~\ref{lem:lcount} implies that there are at least two $p$-elements in $L$.
These imply that \LDeleteMin\ and \LDelete\ 
do not attempt to delete the same element of $L$ in this case. 
On the other hand, if \LDeleteMin\ removes a $p$-element causing $L\_count_p$ to be less than 2
after its completion, 
then it should be that $H_p$ is empty (by Lemma~\ref{lem:lcount}) and thus, it has no element to pull up to the leader-level.
In this case, a future \LDeleteMin\ will still never contend with a \LDelete\ because a thread $p$ will only perform \LDelete\ 
%attempt to delete the same node as a \LDelete\ while $L\_count_p$ is less than 2 because an \LDelete\ will only ever be performed
if $L\_count_p$ is equal to CNTR\_MAX, which must be at least 2.
Thus, in either case, a  \LDeleteMin\ and \LDelete\ (that are executing concurrently) never attempt to delete the same element, as stated in Lemma~\ref{lem:no-interference}. 
This further implies that a node will never be both logically deleted 
\texttt{and} moved according to the relevant marked pointers (Corollary~\ref{cor:marked-both}).
%\pf{This is not accurate since there are also corner cases that have not been discussed above, e.g. if $H_p$ is empty, etc. Can somebody addd some argumentation to cover these cases? 1-2 sentences should be enough.}

\begin{lemma}
\label{lem:no-interference}
An instance of \LDeleteMin\ and an instance of \LDelete\
that are concurrent never attempt to delete the same node in $L$.
\end{lemma}

\begin{corollary}
\label{cor:marked-both} 
A node in $L$ will never have both its DELMIN and its MOVING bit marked.
\end{corollary}

Invariant~\ref{inv:prefix-suffix} states that $L$ is comprised of two parts, 
a prefix that contains all logically-deleted nodes and a suffix 
of non logically-deleted nodes that is sorted.
Lemma~\ref{lem:no-interference} and Corollary~\ref{cor:marked-both} allows us to argue that PIPQ maintains 
the properties of the Lind\'{e}n-Jonsson {\em et al.} algorithm. 
The fact that a \LDelete\ and a \LDeleteMin\ operation never delete the same node (Lemma~\ref{lem:no-interference})
also implies that \LDelete\ respects the properties of Delete from the Harris algorithm. 
These and Lemma~\ref{lem:search properties} allows us to prove the invariant.
The proof of part~\ref{prefix} use arguments from the proof of the invariant that appears in Section $5$ in~\cite{linden}.
The rest of the parts mostly come from arguments needed to prove that the Harris' algorithm is linearizable.

%Some parts of it owe their proof to Lind\'{e}n-Jonsson. 

\begin{invar}
\label{inv:prefix-suffix}
The following claims hold:
\begin{enumerate}
\item \label{prefix} The logically deleted nodes in $L$ form a prefix of the list;
\item \label{suffix} The suffix of $L$ that is comprised of nodes that are not logically deleted is sorted;
\end{enumerate}
\end{invar}

\begin{corollary}
The element with the highest priority in $L$ is its leftomost (first) active node.
\end{corollary}

We are now ready to assign linearization points to the operations supported by $L$.
We linearize a successful \LInsert\ at the point that it successfully executes the CAS
of line~\ref{line:ins_lin}. We linearize a \LDeleteMin\ that deletes a node $v$
at the point that it logically deletes $v$, i.e. at the atomic operation that sets the
DELMIN bit of the previous node to true. We linearize a \LDelete\ at the point
that it sets the MOVING bit of the node it deletes to true.  
Note that no \LInsert\ or \LDelete\ in $L$ may ever be unsuccessful
(since we support duplicates and delete specific elements that we know exist in the list).
An unsuccessful \LDeleteMin\ returns EMPTY and is linearized when the traversal of the list reaches the tail pointer (Line~\ref{line:delmin_empty1} in Algorithm~\ref{alg:l_del-min}). This only may happen if there are not any nodes in the list, or all nodes in the list are marked as logically deleted.
We argue that our implementation of $L$ is linearizable. 
The following lemma is an immediate consequence of how we assigned linearization points
to operations on $L$ (and how the algorithm works). 

\begin{lemma}
\label{lem:L-within}
The linearization point of each \LInsert, \LDeleteMin, or  \LDelete\ 
is within the execution interval of the operation.
\end{lemma}

Let $\ell_{L}$ be the linearization order we derive by the linearization points of the \LInsert, \LDeleteMin, and \LDelete\ operations (of $\alpha$). 
Let $\sigma_{L}$ be the sequential execution that we get by applying the operations of $\ell_{L}$ sequentially
starting from an empty list. 
We say that an operation $op$ on $L$ (in $\alpha$) is {\em consistent} 
with respect to $\ell_{L}$, when the response value of $op$
is the same as that of the corresponding operation in $\sigma_{L}$.

By Lemma~\ref{lem:no-interference}, 
concurrent instances of \LDeleteMin\ and \LDelete\ do not attempt to delete the same node.
By inspection of the pseudocode we see that 
\LDeleteMin\ marks the next field of the last logically deleted node identifying
that its next node, $v$, is now logically-deleted, and returns the key, value and tid of $v$.
Lemma~\ref{lem:no-interference} and Invariant~\ref{inv:prefix-suffix} imply that this is indeed 
the active element with the smallest key in $L$. On the contrary, \LDelete\ always operates
on the non logically-deleted part of the list, i.e. it operates on the sorted suffix of the list
which contains the non logically-deleted nodes.
Lemma~\ref{lem:no-interference} and the facts
that a) \Search, \LInsert\ and \LDelete\ have the same properties as Search, Insert and Delete
in Harris' implementation and b) the way we assign linearization points is similar to that in~\cite{harris},
imply that \LDelete\ is consistent with respect to $\ell_{L}$. Lemma~\ref{lem:L-consistency}
formalizes these claims.

\begin{lemma}
\label{lem:L-consistency}
The response of a \LDeleteMin\ or a \LDelete\ operation is consistent with respect to $\ell_{L}$. 
\end{lemma}

Lemmas~\ref{lem:L-within} and~\ref{lem:L-consistency}
imply that the implementation of $L$  in PIPQ is linearizable.

\begin{corollary}
\label{cor:L-lin}
The implementation of $L$ is linearizable.
\end{corollary}

\noindent
{\bf Linearizability of PIPQ.}
The linearization points of \Insert\ and \DeleteMin\ operations in PIPQ, follow
naturally from the linearization points of the operations on $L$ and $H_p$. 
By Lemma~\ref{lem:hp}(\ref{insert-once}), 
every key that appears as an argument of some \Insert\ operation invoked by $p$ (in $\alpha$) is inserted
either in $L$ or in $H_p$ (not in both).
The linearization point of an \Insert\ operation $op$ by some thread $p$ in PIPQ 
depends on whether the element is inserted in $H_p$ or in $L$. 
In the first case, $op$ is linearized at the point that $p$ aquires the coarse grain lock $l_p$ of $H_p$.
In the case that $op$ inserts the new key in $L$, the linearization point of $op$ is the same as that 
of the \LInsert\ that inserts the key in $L$.
Consider now a \DeleteMin\ operation $op'$ invoked by $p$. By the way PIPQ works, 
there is a single \LDeleteMin\ that is invoked by a combiner $q$ to serve $op'$
and its execution interval is contained in the execution interval of \DeleteMin.
The linearization point of $op'$ is placed at the same place as the linearization point
of this \LDeleteMin.

\begin{lemma}
\label{lem:PIPQ-within}
The linearization point of each \Insert\ and \DeleteMin\ operation $op$ of PIPQ is within the execution interval of $op$.
\end{lemma}

\begin{proof}
Consider first a \Insert\ operation $op$. 
By inspection of the pseudocode of \Insert, it follows that the execution of the single instance of \LInsert\ 
(see Lemma~\ref{lem:hp}(\ref{insert-once})) invoked by $op$ takes place during the execution of $op$. 
Thus, Lemma~\ref{lem:L-within} implies that the claim holds for $op$. 

Consider now a \DeleteMin\ operation $op'$. The linearization point of $op'$ 
is placed at the same place as the linearization point of the unique instance of \LDeleteMin\
invoked by a combiner to execute $op'$. The combining protocol ensures
that each request is served exactly once by some combiner and
that no thread leaves the system before the combiner serves its request. 
Thus, \LDeleteMin\ is executed during the execution interval of $op'$.
Therefore, the claim for $op'$ follows from Lemma~\ref{lem:L-within}.
\end{proof}

Denote by $\ell_H$ the linearization of the worker operations (of $\alpha$).
Recall that these are the operations that are applied on the worker priority queues
(i.e. on every $H_p$). 
Fix any point $t$ in time. Since \WInsert\ and \WDeleteMin\
are linearized at specific instructions of their code,
%we can define the abstract state of $H_p$ at $t$, for any thread $p$,
%as the prefix of $\ell_H$ containing those operations that have been linearized 
%by $t$.
we can define $\ell_H(p,t)$ (for a specific thread $p$),
as the subsequence of $\ell_H$ containing those worker operations that have been linearized 
by $t$ and are applied on $H_p$. By Lemma~\ref{lem:hp}(\ref{hp-lin}),
$H_p$ is linearizable. This allows us to define the abstract state of $H_p$.

\begin{definition}
For each $p$, the abstract state of $H_p$ at $t$ 
is the set of elements contained in the priority queue
that results when the operations in $\ell_H(p,t)$ are applied {\em sequentially}
on an initially empty priority queue.
\end{definition}

%Lemma~\ref{lem:hp}(\ref{}) implies that indeed these are the elements contained in $H_p$ at $t$. 
%\pf{This is not true. We need to change the way we assign linearization points to worker operations.
%The lin point should be at the place where the effect of the operation becomes visible in the data structure.}

Recall that $\ell_L$ the linearization of the operations  (of $\alpha$) that are applied on $L$.
Since \LInsert, \LDeleteMin, and \LDelete\
are linearized at specific instructions of the code,
%we can define the abstract state of $H_p$ at $t$, for any thread $p$,
%as the prefix of $\ell_H$ containing those operations that have been linearized 
%by $t$.
we can define $\ell_L(t)$, 
as the prefix of $\ell_L$ containing those operations on $L$ that have been linearized 
by $t$. By Corollary~\ref{cor:L-lin}, the implementation of $L$ in PIPQ is linearizable. 
This allows us to define the abstract state of $L$.

\begin{definition}
The abstract state of $L$ at $t$ 
is the set of elements contained in the singly-linked list 
that results when the operations in $\ell_L(t)$ are applied {\em sequentially}
on an initially empty linked list.
\end{definition}

\begin{invar}
\label{inv:two-smallest}
For each process $p$, 
the keys of the $p$-elements contained in the abstract state of $L$
are the smallest among all the $p$-elements in the union of the sets that comprise
the abstract state of $H_p$ and the abstract state of $L$. 
\end{invar}

\begin{proof}
Elements tagged by $p$ are only contained in $L$ and $H_p$.
We have to consider all steps that causes the abstract state of $L$ or the abstract state of $H_p$ 
to change and show that after each of these steps, the invariant still holds. 
Let $op$ be the worker operation that executes $s$. 

\begin{enumerate}
\item Assume first that $s$ is a step that causes the abstract state of $H_p$ to change. 
\begin{enumerate}
\item

If $s$ is a step of a \WInsert\ operation by $p$, 
then $s$ is the point where the $l_p$ is acquired by $op$. 
By the way we linearize worker operations, 
$op$ inserts its element $e$ into $H_p$ and is linearized at $s$. 
A \WInsert\ applied on $H_p$ can only be invoked by a \Insert\ by thread $p$. 
Since $e$ is inserted into $H_p$, $p$ will figure out (after $s$ where it acquires $l_p$), 
while still holding $l_p$, that $e$'s key is larger than or equal to the key of the element pointed to by $lead\_largest_p$; note that this is implied if $key$ is larger than or equal to the key of the highest priority element in $H\_p$. 
Lemma~\ref{lem:hp} implies that $lead\_largest_p$ does not change while $op$ is executed.
Moreover, it implies that no $p$-element
will be inserted into $L$ while $op$ holds $l_p$. 
These imply that the execution of $s$ does not violate the invariant.
% This is so because, otherwise, the following would hold: 
% if \Insert\ follows one of the slower paths, it would insert $e$ into $L$ 
% (and therefore not in $H_p$ by Lemma~\ref{lem:hp}(\ref{insert-once}).
% In all other cases (i.e., if the \Insert\ follows either the slower or the slowest path), 
% the key of $e$ is compared with that of the node pointed to by
% $lead\_largest_p$ and only if it is larger, it is inserted in $H_p$. 
% Lemma~\ref{lem:hp} implies that $lead\_largest_p$ does not change while $op$ is executed.
% Moreover, it implies that no $p$-element
% will be inserted into $L$ while $op$ holds $l_p$. 
% These imply that the execution of $s$ does not violate the invariant.
\item If $s$ is a step of a \WDeleteMin\ operation by a thread $q$ applied on $H_p$,
$s$ is again the point where the $l_p$ is acquired by $op$. 
By the way we linearize worker operations, 
$op$ deletes the minimum element $e$ from $H_p$ and is linearized at $s$. 
After $s$, the abstract state of $H_p$ contains an element less, whereas the abstract state of $L$
remains unchanged. Since the invarant holds before $s$, it also holds after $s$. 
\end{enumerate}

\item Assume next that $s$ is a step that causes the abstract state of $L$ to change. 
\begin{enumerate}
\item Assume first that $op$, which executes $s$, is a \LInsert\ that inserts an element $e$ in $L$. 
Then, $s$ is the CAS which connects $e$ into $L$, and the \LInsert\ is linearized at $s$.  
After $s$, the abstract state of $L$ contains an element more, whereas the abstract state of $H_p$
remains unchanged.

\begin{itemize}
\item If $op$ is called by an \Insert\ $op'$ that follows the slower or the slowest path,
Lemma~\ref{lem:hp}(\ref{largest}) implies that $e$ is the element with the highest
priority among all element of $H_p$ at the time that $p$ acquired $l_p$, which was acquired at the beginning of $op'$. 
Since $p$ does not release the lock from that time until  the execution of $s$, the structure of $H_p$ does not change
during this period of time, 
thus $e$ is still an element that has higher priority than all the elements that are contained in $H_p$ at the time
$s$ is executed. Since $p$ holds $l_p$ and does not have an active operation on $H_p$ at the time that $s$ is executed,
$H_p$ is in a quiescent state at $s$. 
Since $H_p$ is linearizable (Lemma~\ref{lem:hp}(\ref{hp-lin})) the elements of this set comprise also the abstract state
of $H_p$ when $s$ is executed. 
These imply that the invariant hold after the execution of $s$.

\item If $op$ is called by a \DeleteMin\ $op'$ then $s$ is executed either by $p$ (through the invocation of $help\_upsert$) 
or by the active combiner $q$ (when it moves up an element to ensure that $L\_count_p$ is greater than 1). 
Operation $op'$ acquires $l_p$ while $op$ is being executed. 
Lemma~\ref{lem:hp}(\ref{largest}) implies that $e$ is the element with the highest
priority among all element of $H_p$ at the time $t'$ that $p$ acquired $l_p$. 
Similarly to the previous case, since $p$ does not release the lock from $t'$ until  the execution of $s$, 
the structure of $H_p$ does not change during this period of time, 
thus $e$ is still an element that has higher priority than all the elements that are contained in $H_p$ at the time
$s$ is executed. Since $p$ holds $l_p$ and does not have an active operation on $H_p$ at the time that $s$ is executed,
$H_p$ is in a quiescent state at $s$. 
Since $H_p$ is linearizable (Lemma~\ref{lem:hp}(\ref{hp-lin})) the elements of this set comprise also the abstract state
of $H_p$ when $s$ is executed. 
These imply that the invariant hold after the execution of $s$.
\end{itemize}

\item Assume next that $op$, which executes $s$, is a \LDelete\ or \LDeleteMin\ that deletes an element $e$ from $L$. 
By the way linearization points are assigned to \LDeleteMin\ and \LDelete, 
$s$ marks the LOGDEL bit or MOVING bit of $e$, respectively, and $op$ is linearized at $s$. 
After $s$, the abstract state of $L$ contains an element less, whereas the abstract state of $H_p$
remains unchanged. Since the invariant holds before the execution of $s$, it also holds after it.
\end{enumerate}
\end{enumerate}
\end{proof}

We are now ready to define the abstract state of the priority queue implemented by PIPQ
(which is a multi-level data structure consisting of all the worker-level min-heaps and $L$).

\begin{definition}
At each point $t$, the abstract state of PIPQ is defined 
by the (multi-set) union of the abstract states of $H_1$, $H_2$, ..., $H_n$ and the abstract state of $L$.
\end{definition}

Denote by $\ell$ the linearization order we derive by the linearization points 
of the \Insert\ and \DeleteMin\ operations (of $\alpha$). 
Let $\sigma$ be the sequential execution that we get by applying the operations of $\ell$ sequentially
starting from an empty priority queue. 
We say that an \Insert\ or a \DeleteMin\ operation $op$ in $\alpha$ is {\em consistent} 
with respect to $\ell$, if the response value of $op$
is the same as that of the corresponding operation in $\sigma$.
Invariant~\ref{inv:two-smallest} and Lemma~\ref{lem:PIPQ-within} imply that the response of each \DeleteMin\ operation 608462469 / 12275954

is consistent with respect to $\ell$, as stated in Lemma~\ref{lem:consistency}. 

\begin{lemma}
\label{lem:consistency}
The responses of \DeleteMin\ operations in PIPQ are consistent with respect to $\ell$. 
\end{lemma}

Lemma~\ref{lem:consistency} implies that PIPQ is linearizable. 

\begin{theorem}
PIPQ is a linearizable implementation of a priority queue. 
\end{theorem}

%-----------------------------------------
% LINEARIZATION PTS:
%
% INS:
% Fast path: once local heap's lock is acquired (or after released?) 
% Slow path: same as lin point of L-Insert: linearized on successful CAS
% Slowest path: same as lin point of L-DeleteMaxP: linearized on successful CAS (the rest of the algorithm is just moving an element from leader -> worker, in such a way that another thread could not "miss it" somehow)
%
% DEL_MIN:
% Linearization point is same as L-DELMIN: lin pt is successful marking as logically deleted of the first unmarked next pointer via the fetch-and-or operation
%-----------------------------------------